%% file: arxiv.tex
\documentclass{article}
\usepackage[utf8]{inputenc}
\usepackage{mathtools}
\usepackage{macros}
\usepackage[margin=1in]{geometry}

\newtheorem{theorem}{Theorem}[section]
\newtheorem{lemma}[theorem]{Lemma}
\newtheorem{corollary}[theorem]{Corollary}
\newtheorem{definition}[theorem]{Definition}

\newtheorem{question}[theorem]{Question}

\usepackage[backref,
            colorlinks = true,
            linkcolor = blue,
            citecolor = gray,
            backref]{hyperref}
\hypersetup{
    linktoc=all, %
}
\usepackage{footnotebackref}
\usepackage[capitalise]{cleveref}

\usepackage{graphicx}
\graphicspath{ {./images/} }

\newcommand{\frS}{\mathfrak{S}}

\title{Sharper Bounds for $\ell_p$ Sensitivity Sampling\thanks{Our sensitivity sampling results for $p\in[1, 2)$ are subsumed by an earlier sharper result of Chen--Derezinski \cite{CD2021}, which shows that the sensitivity scores upper bound Lewis weights up to a factor of $d^{1-p/2}$, which implies a sampling complexity bound of $\tilde O(\eps^{-2}d^{1-p/2}\mathfrak S)$. We refer to Section \ref{sec:contributions} as well as the recent work of \cite{MO2023} for more details.}}
\author{
David P. Woodruff \\ Carnegie Mellon University \\ \texttt{dwoodruf@cs.cmu.edu} \and Taisuke Yasuda \\ Carnegie Mellon University \\ \texttt{taisukey@cs.cmu.edu}
}
\date{}

\begin{document}

\maketitle

\thispagestyle{empty}
\begin{abstract}
\input{abstract.tex}
\end{abstract}

\clearpage
\setcounter{page}{1}

\input{paper.tex}

\bibliographystyle{alpha}
\bibliography{example_paper}

\input{appendix.tex}

\appendix

\end{document}

%% file: abstract.tex
In large scale machine learning, \emph{random sampling} is a popular way to approximate datasets by a small representative subset of examples. In particular, \emph{sensitivity sampling} is an intensely studied technique which provides provable guarantees on the quality of approximation, while reducing the number of examples to the product of the \emph{VC dimension} $d$ and the \emph{total sensitivity} $\frS$ in remarkably general settings. However, guarantees going beyond this general bound of $\frS d$ are known in perhaps only one setting, for \emph{$\ell_2$ subspace embeddings}, despite intense study of sensitivity sampling in prior work. In this work, we show the first bounds for sensitivity sampling for $\ell_p$ subspace embeddings for $p > 2$ that improve over the general $\frS d$ bound, achieving a bound of roughly $\frS^{2-2/p}$ for $2<p<\infty$. Furthermore, our techniques yield further new results in the study of sampling algorithms, showing that the \emph{root leverage score sampling} algorithm achieves a bound of roughly $d$ for $1\leq p<2$, and that a combination of leverage score and sensitivity sampling achieves an improved bound of roughly $d^{2/p}\frS^{2-4/p}$ for $2<p<\infty$. Our sensitivity sampling results yield the best known sample complexity for a wide class of structured matrices that have small $\ell_p$ sensitivity.

%% file: paper.tex
\section{Introduction}

In typical large scale machine learning problems, one encounters a dataset represented by an $n\times d$ matrix $\bfA$, consisting of $d$ features and a large number $n\gg d$ of training examples. While an extremely large $n$ can cause various data analytic tasks to be intractable, it is often the case that not all $n$ training examples are necessary, and \emph{random sampling} can effectively reduce the number of training examples while approximately preserving the information necessary for downstream prediction tasks. 

Uniform sampling is perhaps the simplest instance of this idea that is often used in practice. However, uniform sampling can lead to significant information loss when there are a small number of important training examples that must be kept under sampling. Thus, recent work, both in theory \cite{LS2010, FL2011} and in practice \cite{KF2017, JG2018}, has focused on \emph{importance sampling} methods that sample more important examples with higher probability.

In this work, we focus on the use of importance sampling techniques for approximating objective functions for empirical risk minimization problems. Consider an objective function $f:X \to \mathbb R_{\geq 0}$ of the form of a sum along the coordinates, i.e.,
\[
    f(\bfx) = \sum_{i=1}^n f_i(\bfx),
\]
where $X$ is some domain set and $f_i:X\to\mathbb R_{\geq 0}$ are non-negative loss functions for $i\in[n] = \{1, 2, \ldots, n\}$. Then, we seek algorithms that sample a subset $S\subseteq[n]$ and weights $\bfw_i$ for $i\in S$ such that, with high probability, the function
\[
    \tilde f(\bfx) \coloneqq \sum_{i\in S}\bfw_i f_i(\bfx)
\]
satisfies 
\begin{equation}\label{eq:rel-err}
\mbox{for every $\bfx\in X$, }\qquad\tilde f(\bfx) = (1\pm\eps) f(\bfx)
\end{equation}
for some accuracy parameter $0<\eps<1$. We will also refer to the $\eps$ parameter as the \emph{sampling error}. Note then that if $\abs{S} \ll n$, then $\tilde f$ can be used as a surrogate objective function in downstream applications that can be processed much more efficiently than $f$ itself.

\paragraph{Sensitivity Sampling.}

The \emph{sensitivity sampling} framework, introduced by \cite{LS2010, FL2011}, provides one method of achieving guarantees of the form of \eqref{eq:rel-err}. In this framework, one first computes\footnote{While computing sensitivity scores exactly is often inefficient, efficiently computable approximations exist and suffice in many cases.} \emph{sensitivity scores} of each coordinate $i\in[n]$:
\[
    \bfsigma_i \coloneqq \sup_{\bfx\in X} \frac{f_i(\bfx)}{\sum_{j=1}^n f_j(\bfx)}.\footnote{We define the fraction to be $0$ when the denominator is $0$.}
\]
Then, each $i\in[n]$ is independently sampled with probability $p_i$ proportional to $\bfsigma_i$, and the sampled row is assigned a weight of $1/p_i$. It is easy to see that this preserves the objective function $f(\bfx)$ in expectation for every $\bfx\in X$. Furthermore, it can be shown that sampling $\tilde O(\eps^{-2}\frS d)$ rows provides a $(1\pm\eps)$-factor approximation to the objective function \cite{BFL2016, FSS2020}, where $\frS = \sum_{i=1}^n \bfsigma_i$ is known as the \emph{total sensitivity} and $d$ is the VC dimension of an associated set system.

\paragraph{Sampling Algorithms for $\ell_p$ Linear Regression.}

We now turn to sampling algorithms for \emph{$\ell_p$ linear regression}, which is the main problem of study of this work. Consider an input matrix $\bfA\in\mathbb R^{n\times d}$ with $n$ rows $\bfa_i\in\mathbb R^d$, a label vector $\bfb\in\mathbb R^n$, and $1\leq p < \infty$. We then seek to minimize
\[
    f(\bfx) = \sum_{i=1}^n \abs*{\angle*{\bfa_i, \bfx} - \bfb_i}^p = \norm*{\bfA\bfx - \bfb}_p^p.
\]
Note that this problem is in the form of an empirical risk minimization problem as discussed previously, with $f_i(\bfx) = \abs*{\angle*{\bfa_i, \bfx} - \bfb_i}^p$ and $X = \mathbb R^d$. Furthermore, in the case of $\ell_p$ regression, we may use the scale invariance of the $\ell_p$ norm to fold the weights $\bfw_i$ into the objective $f_i$, so we can write the sampling procedure as a diagonal map:

\begin{definition}[$\ell_p$ Sampling Matrix]\label{def:sampling-matrix}
Let $1\leq p < \infty$. A random diagonal matrix $\bfS\in\mathbb R^{n\times n}$ is a \emph{random $\ell_p$ sampling matrix with sampling probabilities $\{q_i\}_{i=1}^n$} if for each $i\in[n]$, the $i$th diagonal entry is independently set to be
\[
    \bfS_{i,i} = \begin{cases}
        1 / q_i^{1/p} & \text{with probability $q_i$} \\
        0 & \text{otherwise}
    \end{cases}
\]
\end{definition}

Our goal then is to compute probabilities $\{p_i\}_{i=1}^n$ such that the associated $\ell_p$ random sampling matrix $\bfS$ satisfies
\[
    \mbox{for every $\bfx\in\mathbb R^d$,} \quad\norm*{\bfA\bfx - \bfb}_p^p = (1\pm\eps)\norm*{\bfS\bfA\bfx - \bfS\bfb}_p^p
\]
Note that we may in fact also assume that $\bfb = 0$, since we can append $\bfb$ to be one of the columns of $\bfA$. Thus, our problem is to compute probabilities $\{p_i\}_{i=1}^n$ satisfying
\begin{equation}\label{eq:se}
\mbox{for every $\bfx\in\mathbb R^d$,} \quad\norm*{\bfS\bfA\bfx}_p^p = (1\pm\eps)\norm*{\bfA\bfx}_p^p,
\end{equation}
which is known as an \emph{$\ell_p$ subspace embedding} of $\bfA$.

We introduce the following notation for sensitivity sampling when specifically applied to $\ell_p$ subspace embeddings:

\begin{definition}[$\ell_p$ sensitivities]
Let $\bfA\in\mathbb R^{n\times d}$ and $p \geq 1$. Then, for each $i\in[n]$, we define the \emph{$i$th $\ell_p$ sensitivity} to be
\[
    \bfsigma_i^p(\bfA) \coloneqq \sup_{\bfx\in\mathbb R^d, \bfA\bfx\neq 0}\frac{\abs{[\bfA\bfx](i)}^p}{\norm*{\bfA\bfx}_{p}^p}\footnote{We write $[\bfA\bfx](i)$ for the $i$th entry of the vector $\bfA\bfx\in\mathbb R^n$.}
\]
and the \emph{total $\ell_p$ sensitivity} to be $\frS^p(\bfA) \coloneqq \sum_{i=1}^n \bfsigma_i^p(\bfA)$.
\end{definition}

Note that the calculation of $\ell_p$ sensitivities can be formulated as an $\ell_p$ regression problem, and can be computed efficiently using recent developments in algorithms for $\ell_p$ regression. Indeed, it is easy to see that
\[
    \frac1{\bfsigma_i^p(\bfA)} = \min_{[\bfA\bfx](i) = 1} \norm*{\bfA\bfx}_p^p,
\]
which can be efficiently approximated to high precision in nearly matrix multiplication time \cite{AKPS2019, APS2019, AS2020}.

For $p = 2$, the $\ell_p$ sensitivities are exactly equal to the \emph{leverage scores}.

\begin{definition}[Leverage scores]
\label{def:lev-scores}
Let $\bfA\in\mathbb R^{n\times d}$. Then, for each $i\in[n]$, we define the \emph{$i$th leverage score} to be $\bftau_i(\bfA) \coloneqq \bfsigma_i^2(\bfA)$.
\end{definition}

For $\ell_p$ subspace embeddings, it is known that the total sensitivity is at most $d$ for $p \leq 2$ and at most $d^{p/2}$ for $p>2$, where $d$ is the dimension. Thus, sensitivity sampling applies in this setting, and has indeed been successfully applied in prior work \cite{BDMMUWZ2020,BHMSSZ2021}. There are also other sampling schemes, based on Lewis weights, which we discuss more below.

For $p = 2$, which corresponds to the standard least squares regression problem, it has long been known that sensitivity sampling, known as \emph{leverage score sampling} in this case, yields $\ell_2$ subspace embeddings \eqref{eq:se} with nearly optimal sample complexity of\footnote{We write $\tilde O(f)$ to denote $f\poly\log f$.} $\tilde O(\eps^{-2} d)$ \cite{Mah2011}. However, it is also known that the total sensitivity $\frS$ is exactly $d$ in this setting, and the dimension is also $d$. Thus, in the natural setting of $\ell_2$ subspace embeddings, the bound of $\frS d = d^2$ for sensitivity sampling is quadratically loose in the analysis. However, to the best of our knowledge, no bound which improves over the general bound of $\frS d$ specifically for sensitivity sampling is known in any other setting. We thus arrive at the central question of our work:

\begin{question}\label{q:main}\leavevmode
\begin{center} \it
How many samples are necessary for sensitivity sampling to output an $\ell_p$ subspace embedding \eqref{eq:se}?
\end{center}
\end{question}

\subsection{Our Contributions}
\label{sec:contributions}

For $1\leq p < 2$, an improved upper bound on the sample complexity of $\ell_p$ sensitivity sampling is implicit in the work of Chen and Derezinski \cite{CD2021}. For this range of $p$, this work shows that the sensitivity scores are related up to a $d^{1-p/2}$ factor to \emph{Lewis weights} \cite{Lew1978, BLM1989, LT1991, CP2015, WY2023}, which are sampling scores known to achieve a sample complexity of $\tilde O(\eps^{-2}d)$ for the problem of sampling $\ell_p$ subspace embeddings for $1\leq p < 2$. Then by using sampling probabilities that are a $d^{1-p/2}$ factor larger than the sensitivity scores, we can sample rows with probability at least the Lewis weights, which implies a sampling complexity of $\tilde O(\eps^{-2}\mathfrak S d^{1-p/2})$ for sensitivity sampling. The regime of $1 \leq p < 2$ is, however, perhaps less interesting, as Lewis weight sampling already achieves nearly optimal bounds for any rank $d$ matrix. For the more interesting case of $2 < p < \infty$, such a relationship between the Lewis weights and the sensitivity scores is not known, and furthermore, relating to Lewis weights is not as useful since the sample complexity for Lewis weight sampling scales as $\tilde O(\eps^{-2}d^{p/2})$, even if the total sensitivity $\mathfrak S$ is much less than $d^{p/2}$.

In this work, we make progress towards resolving Question~\ref{q:main} by obtaining the first bounds for sensitivity sampling that go beyond the general $\frS d$ bound for $2 < p < \infty$.

\begin{theorem}[Informal Restatement of \cref{thm:sens-sample-p<2} and \cref{thm:sens-sample-p>2}]\label{thm:main-informal}
Let $1 \leq p < \infty$ and let $\bfA\in\mathbb R^{n\times d}$. Let $\alpha>0$ and let $q_i = \min\{1, 1/n + \bfsigma_i^p(\bfA) / \alpha\}$ for $i\in[n]$. Then, there is an $\alpha$ such that the random $\ell_p$ sampling matrix $\bfS$ with sampling probabilities $\{q_i\}_{i=1}^n$ satisfies \eqref{eq:se} and samples at most $m$ rows with probability at least $1-1/\poly(n)$, with
\[
    m = \begin{dcases}
        \frac{\frS^p(\bfA)^{2/p}}{\eps^2}\poly\log n & 1 \leq p < 2 \\
        \frac{\frS^p(\bfA)^{2-2/p}}{\eps^2}\poly\log n & 2 < p < \infty
    \end{dcases}
\]
\end{theorem}

For $1\leq p < 2$, our techniques give a slightly weaker bound of $\eps^{-2}\mathfrak S^p(\bfA)^{2/p}\poly\log n$, which is looser than the bound of $\eps^{-2}d^{1-2/p}\mathfrak S^p(\bfA)\poly\log n$ obtained by \cite{CD2021}. On the other hand, our techniques completely avoid the notion of Lewis weights, and thus may prove to be useful in other sensitivity sampling settings, where it is more difficult to define analogous weights. 

For $1\leq p < 2$, we show that our bound of \cref{thm:main-informal} (as well as that of \cite{CD2021}) is in fact tight, in the sense that there exist matrices such that, up to logarithmic factors, $\frS^p(\bfA)^{2/p}$ samples are necessary to obtain $\ell_p$ subspace embeddings. We show this by showing that random matrices have a total sensitivity of at most $\frS^p(\bfA) = \tilde O(d^{p/2})$, which implies the claim since $d$ rows are necessary to even maintain the rank.

\begin{theorem}
\label{thm:small-total-sens}
Let $1\leq p < 2$. Let $n = d^{O(p)}$ be large enough, and let $\bfA$ be a random $n\times d$ standard Gaussian matrix. Then, with probability at least $2/3$, $\frS^p(\bfA) \leq O(d\log d)^{p/2}$.
\end{theorem}

Furthermore, we show the first bounds showing that the total $\ell_p$ sensitivity cannot be any smaller than the construction in \cref{thm:small-total-sens}, up to logarithmic factors:

\begin{theorem}
\label{thm:total-sens-lb}
Let $\bfA\in\mathbb R^{n\times d}$ and $1 \leq p < \infty$. Then, 
\[
    \frS^p(\bfA) \geq \begin{cases}
        d/2 & p > 2 \\
        d^{p/2} / 2 & p < 2
    \end{cases}
\]
\end{theorem}

\paragraph{Comparison to $\ell_p$ Lewis Weight Sampling.}

We now compare our results to the well-known \emph{$\ell_p$ Lewis weight sampling} technique, which yields nearly optimal bounds for $\ell_p$ subspace embeddings for matrices with worst-case total $\ell_p$ sensitivity \cite{CP2015, WY2023}. More specifically, it is known that the total $\ell_p$ sensitivity is at most $d$ for $p<2$ and at most $d^{p/2}$ for $p>2$ for all matrices $\bfA\in\mathbb R^{n\times d}$, and $\ell_p$ Lewis weight sampling provides a way to reduce the number of rows to at most roughly $d^{1\lor (p/2)}$ rows\footnote{We write $a\lor b$ to denote the maximum of $a$ and $b$.} for all matrices $\bfA$:

\begin{theorem}[$\ell_p$ Lewis weight sampling]\label{thm:lewis-weight-sampling}
Let $\bfA\in\mathbb R^{n\times d}$ and $p>0$. Let $\bfS$ be a random $\ell_p$ sampling matrix with sampling probabilities $\{q_i\}_{i=1}^n$ proportional to the $\ell_p$ Lewis weights. Then, $\bfS$ samples $O(\eps^{-2}d^{1\lor(p/2)}(\log d)^2\log(d/\eps))$ rows and satisfies \eqref{eq:se}.
\end{theorem}

In particular, $\ell_p$ Lewis weight sampling yields better sampling complexity bounds than our \cref{thm:main-informal} when $p<2$ or when the $\ell_p$ total sensitivity $\frS^p(\bfA)$ is close to $d^{p/2}$. However, our results on $\ell_p$ sensitivity sampling have a number of advantages over the existing $\ell_p$ Lewis weight sampling results. First, Lewis weight sampling, to the best of our knowledge, is not known to admit bounds better than $d^{p/2}$ for matrices with very small total sensitivity $\frS^p(\bfA) \ll d^{p/2}$ for $p>2$. Thus, for matrices with substantially smaller total sensitivity, $\ell_p$ sensitivity sampling yields far improved bounds. We illustrate a number of such applications below. Second, sensitivity sampling 
generalizes to sampling problems beyond $\ell_p$ subspace embeddings and has been applied %
to subspace embeddings for general $M$-estimators and Orlicz norms \cite{CW2015a, CW2015b, TMF2020, MMWY2022}, logistic regression \cite{MSSW2018} and other generalized linear models \cite{MOP2022}, as well as general shape fitting problems including clustering and subspace approximation \cite{HV2020} and projective clustering \cite{VX2012}. Thus, our techniques may be useful for improving the analyses of a broad range of sampling problems.

\paragraph{Other Sampling Algorithms.}

In addition to $\ell_p$ sensitivity sampling, our techniques yield other new results on sampling algorithms for $\ell_p$ subspace embeddings. 

Our next result is a new analysis of \emph{root leverage score sampling}, which is a popular method for efficiently computing upper bounds to sensitivity scores for loss functions of at most quadratic growth, including $\ell_p$ losses for $1\leq p < 2$, the Huber loss, and the logistic loss \cite{CW2015b, MSSW2018, GPV2021}. In this technique, the $p_i$ are set proportionally to the \emph{square root} of the $\ell_2$ leverage scores (\cref{def:lev-scores}). While the sum of the root leverage scores can be as large as $\sqrt{nd}$, this sampling procedure can be recursively applied for $O(\log\log n)$ iterations to reduce the sample complexity to $\poly(d)$.

As with sensitivity sampling, the only previously known analyses of root leverage score sampling proceed by a na\"ive union bound, which can only reduce the number of samples to $d^2$, which is loose by a $d$ factor. Our techniques for analyzing sensitivity sampling can be modified to show that root leverage score sampling in fact leads to a \emph{nearly optimal} number of samples:

\begin{theorem}[Informal Restatement of \cref{thm:root-lev-sampling} and \cref{thm:recursive-root-lev-sampling}]
\label{thm:informal-root-lev}
Let $1\leq p < 2$ and let $\bfA\in\mathbb R^{n\times d}$. Let $\alpha>0$ and let $p_i = \min\{1,\bftau_i(\bfA)^{p/2}/\alpha\}$ for $i\in[n]$. Then, there is an $\alpha$ such that the random $\ell_p$ sampling matrix $\bfS$ with sampling probabilities $\{p_i\}_{i=1}^n$ satisfies \eqref{eq:se} and samples at most $m$ rows with probability at least $1 - 1/\poly(n)$, with
\[
    m = \frac{n^{1-p/2}d^{p/2}}{\eps^2}\poly\log n.
\]
Recursively applying this result gives a matrix $\bfS$ satisfying \eqref{eq:se} with at most $m$ rows, for
\[
    m = \frac{d}{\eps^{4/p}}\poly\log n.
\]
\end{theorem}

Our analysis of root leverage score sampling provides a promising direction towards resolving the problem of designing nearly optimal sampling algorithms for preserving subspaces under the \emph{Huber loss}, which has been raised as an important question on sampling algorithms in a number of works \cite{AS2020, GPV2021, MMWY2022} for its applications in Huber regression and fast algorithms for $\ell_p$ regression. Here, the best known upper bound is a sampling algorithm reducing the number of rows to roughly\footnote{A polylogarithmic dependence on $n$ is necessary and sufficient here, but we omit this from our discussion for simplicity.} $d^{4-2\sqrt 2} \approx d^{1.172}$, whereas $d$ is conjectured to be possible \cite{MMWY2022}. Root leverage scores are known to upper bound the sensitivities for the Huber loss \cite{CW2015a, CW2015b, GPV2021}, so our \cref{thm:informal-root-lev} suggests that root leverage score sampling may yield a sampling algorithm reducing the number of samples to $d$ for the Huber loss as well. 

We additionally show that by incorporating $\ell_2$ leverage scores into $\ell_p$ sensitivity sampling, we can obtain sampling guarantees that further improve over the guarantee of \cref{thm:main-informal} for $p>2$. We note that our proof of this result uses a recursive ``flattening and sampling'' scheme for this result, rather than a direct sampling result as in our earlier results.

\begin{theorem}[Informal Restatement of \cref{thm:recursive-sens-lev-sampling}]
\label{thm:sens-lev-informal}
Let $2 < p < \infty$ and let $\bfA\in\mathbb R^{n\times d}$. Then, there is an efficient algorithm which computes a matrix $\bfS\in\mathbb R^{m\times n}$ satisfying \eqref{eq:se}, for
\[
    m = \frac{d^{2/p}\frS^p(\bfA)^{1-2/p}}{\eps^2}\poly\log n.
\]
\end{theorem}

Although \cref{thm:sens-lev-informal} does not specifically use sensitivity sampling, to the best of our knowledge, it is the best known sampling result for $\ell_p$ subspace embeddings with small $\ell_p$ total sensitivity $\frS^p(\bfA)$.

\paragraph{Applications.}

We now show several examples in structured regression problems in which our new sensitivity sampling results give the best known sample complexity results for $\ell_p$ subspace embeddings for $p>2$. We start by presenting a couple of lemmas which show that certain natural classes of matrices have total $\ell_p$ sensitivity  $\ll d^{p/2}$.

The first result is a lemma extracted from a result of \cite{MMMWZ2022} bounding the total $\ell_p$ sensitivity for a sparse perturbation of low rank matrices:

\begin{lemma}[Sensitivity Bounds for Low Rank + Sparse Matrices \cite{MMMWZ2022}]
\label{lem:sens-low-rank-sparse}
Let $\bfA = \bfK + \bfS \in\mathbb R^{n\times d}$ for a rank $k$ matrix $\bfK$ and an $\bfS$ with at most $s$ nonzero entries per row. Let $1 \leq p < \infty$. Then, $\frS^p(\bfA) \leq d^s (k+s)^p$. 
\end{lemma}

We provide a self-contained proof in \cref{sec:structured-mat-sens}.

In a second example, we show that ``concatenated Vandermonde'' matrices, which were studied in, e.g., 
\cite{ASW2013}, also have small total $\ell_p$ sensitivity. These matrices naturally arise as the result of applying a polynomial feature map to a matrix.

\begin{definition}[Vandermonde matrix]
Given a vector $\bfa\in\mathbb R^n$, the degree $q$ Vandermonde matrix $V^q(\bfa)\in\mathbb R^{n\times (q+1)}$ is defined entrywise as $V^q(\bfa)_{i,j} = \bfa_i^j$ for $j = 0, 1, \dots, q$.
\end{definition}

\begin{definition}[Polynomial feature map]
Given a matrix $\bfA\in\mathbb R^{n\times k}$ and an integer $q$, we define the matrix $V^q(\bfA)\in\mathbb R^{n\times k(q+1)}$ to be the horizontal concatenation of the Vandermonde matrices $V^q(\bfA\bfe_1), V^q(\bfA\bfe_2), \dots, V^q(\bfA\bfe_k)$.
\end{definition}

We show the following result, proven in \cref{sec:structured-mat-sens}.

\begin{lemma}[Sensitivity Bounds for Matrices Under Polynomial Feature Maps]
\label{lem:poly-feature-map-sens}
Let $\bfA\in\mathbb R^{n\times k}$ and let $q$ be an integer. Let $1\leq p < \infty$. Then, $\frS^p(V^q(\bfA)) \leq (pq+1)^k$. 
\end{lemma}

This generalizes a result of \cite{MMMWZ2022}, which bounds the $\ell_p$ sensitivities of a single Vandermonde matrix.

In the low-sensitivity matrices of \cref{lem:sens-low-rank-sparse} and \cref{lem:poly-feature-map-sens}, it is in fact possible to apply Lewis weight sampling to obtain sampling bounds that match these sensitivity bounds, by using the \emph{tensoring trick} \cite{MMMWZ2022}. However, when a tiny amount of noise is added to these matrices, then algebraic tricks such as tensoring break down, and the sensitivity bounds derived from Lewis weights increase substantially to $d^{p/2}$ for $p>2$. On the other hand, sensitivity sampling itself is robust with respect to the addition of noise, as it depends only on norms rather than brittle quantities such as rank. Indeed, we have the following fact, which we prove in \cref{sec:sens-noisy-matrix}:

\begin{lemma}
\label{lem:sens-noisy-matrix}
Let $\bfA\in\mathbb R^{n\times d}$ be a rank $d$ matrix with minimum singular value $\sigma_{\min}$. Let $\bfE\in\mathbb R^{n\times d}$ be an arbitrary perturbation matrix with
\[
    \norm*{\bfE}_2 \leq \frac{\sigma_{\min}}{2n^{1+1/p}}.
\]
Then, $\frS^p(\bfA+\bfE) \leq 2^p(\frS^p(\bfA) + 1)$.
\end{lemma}

Thus, for small perturbations of structured matrices with small $\ell_p$ sensitivity as specified by \cref{lem:sens-noisy-matrix}, \cref{thm:main-informal} and \cref{thm:sens-lev-informal} give the tightest known bounds on the sample complexity for $\ell_p$ subspace embeddings. Such perturbations may arise due to roundoff error or finite precision on a computer, and no prior bounds beating Lewis weight sampling or the na\"ive $\frS d$ bound for sensitivity sampling were known for the applications above.

\subsection{Other Related Work}

The problem of designing sampling algorithms for $\ell_p$ subspace embeddings has a long and rich history, dating back to works in the functional analysis literature and culminating in the Lewis weight sampling result \cite{Lew1978, Sch1987, BLM1989, Tal1990, LT1991, Tal1995, SZ2001}. More recently, the theoretical computer science community has studied this problem for its applications to $\ell_p$ linear regression and other empirical risk minimization problems. The early works of \cite{Cla2005, DDHKM2009} obtained sampling algorithms for $\ell_p$ regression based on sensitivity score upper bounds given by various constructions of \emph{$\ell_p$ well-conditioned bases} for $\bfA$. The theory of Lewis weight sampling was brought to the theoretical computer science literature by \cite{CP2015}, and has been improved both in the computation of the weights \cite{Lee2016, FLPS2022, JLS2022} and sampling guarantees \cite{WY2022, WY2023}.

\section{Sampling Error Bounds}
\label{sec:sample-error-bound}

We now give a more detailed discussion of our techniques and results. While our discussion in this section will present most of the major ideas necessary to prove our new results, all of the full proofs will be deferred to the appendix due to space considerations.

\subsection{Prior Approaches}

We start by describing the standard proof of the $\tilde O(\eps^{-2}\frS d)$ sample complexity bound for sensitivity sampling \cite{Sch1987}. Using Bernstein bounds, it can be shown that sampling $O(\eps^{-2}\frS\log\frac1\delta)$ rows preserves the $\ell_p$ norm of a fixed vector $\bfA\bfx$ up to a $(1\pm\eps)$ factor with probability at least $1-\delta$. One can then consider an \emph{$\eps$-net} $N$, which is a set of size roughly $1/\eps^d$ such that any $\ell_p$ unit vector $\bfA\bfx$ is $\eps$-close to some $\bfA\bfx'\in N$. By a union bound, the norm preservation guarantee holds simultaneously for every $\bfA\bfx'\in N$ with constant probability, if we set $\delta = \eps^d$. Now for an arbitrary $\ell_p$ unit vector $\bfA\bfx$, the norm preservation guarantee holds for an $\eps$-close point $\bfA\bfx'\in N$, which implies the norm preservation guarantee for $\bfA\bfx$ itself by a standard argument. Finally, scale invariance ensures that the same conclusion holds for all vectors $\bfA\bfx$, rather than just unit vectors.

To improve over this argument, the $\ell_p$ Lewis weight sampling technique was developed in a line of work from the functional analysis literature \cite{Lew1978, BLM1989, Tal1990, LT1991, Tal1995}, which incorporates \emph{chaining arguments}. Chaining arguments are a way of improving $\eps$-net arguments by using a \emph{sequence} of $\eps$-nets at different ``scales'', rather than using a single scale of $\eps$, and using tighter bounds for net constructions (i.e., smaller cardinality nets) at larger scales (see, e.g., \cite{Nel2016} for a survey of chaining applications in computer science). 

We now delve into a discussion of the overall strategy towards bounding the \emph{sampling error} of our $\ell_p$ sampling results, which is a random variable $\Lambda$ depending on $\bfS$ given by
\begin{equation}\label{eq:sampling-error}
    \Lambda \coloneqq \sup_{\norm*{\bfA\bfx}_p = 1}\abs*{\norm*{\bfS\bfA\bfx}_p^p - 1}.
\end{equation}
Note that if $\Lambda \leq \eps$, then $\norm*{\bfS\bfA\bfx}_p^p = (1\pm\eps)\norm*{\bfA\bfx}_p^p$ for every $\bfx\in\mathbb R^d$. In our discussions in this section, we will focus on bounding $\Lambda$ in expectation, although our full proofs in the appendix will bound higher moments of $\Lambda$ to obtain high probability bounds.

\subsection{Generalized Chaining Bounds for \texorpdfstring{$\ell_p$}{lp} Subspace Embeddings}

Our main technical lemma towards bounding \eqref{eq:sampling-error} is a generalization of the chaining argument framework for Lewis weight sampling, which bounds the sampling error of $\ell_p$ sampling algorithms by the leverage scores and $\ell_p$ sensitivities of the sampled matrix $\bfS\bfA$, rather than by the Lewis weights of $\bfS\bfA$.

First, we introduce our sampling bounds obtained by generalizing the chaining arguments of \cite{BLM1989, LT1991}. In this result, we obtain the following bound on a certain Rademacher process, which can be interpreted as the sampling error of a uniform sampling process, as we describe in \cref{sec:gaussianization-reduction}.
\begin{lemma}[Rademacher Process Bound, Simplified]\label{lem:rad-bound}
Let $\bfA\in\mathbb R^{n\times d}$ and $1\leq p < \infty$. Let $\tau\geq\bftau_i(\bfA)$ and $\sigma\geq\bfsigma_i^p(\bfA)$ for every $i\in[n]$. Let
\[
    E\coloneqq \E_{\bfeps\sim\{\pm1\}^n}\sup_{\norm*{\bfA\bfx}_p=1}\abs*{\sum_{i=1}^n \bfeps_i \abs*{[\bfA\bfx](i)}^p}.
\]
Then,
\[
    E \leq \begin{cases}
    O(\tau^{1/2})(\log n)^{3/2} & p < 2 \\
    O(\tau^{1/2})(\sigma n)^{1/2-1/p}(\log n)^{3/2} & p > 2
    \end{cases}
\]
\end{lemma}
This result follows from a Gaussianization argument (\cref{lem:gp-reduction}), followed by an application of Dudley's entropy integral theorem (\cref{thm:dudley}), and then bounding the entropy integral in \cref{lem:entropy-int-p<2} and \cref{lem:entropy-int-p>2}. 

As we discuss in \cref{sec:gaussianization-reduction}, \cref{lem:rad-bound} will be a key ingredient in bounding the sampling error in our sensitivity sampling analysis.

\subsection{Leverage Score and Sensitivity Bounds}

In our final argument, \cref{lem:rad-bound} will be applied with the matrix $\bfA$ set to be (a modified version of) the sampled matrix $\bfS\bfA$. Thus, we require a bound on the leverage scores and sensitivities $\tau$ and $\sigma$ of $\bfS\bfA$. In fact, $\sigma$ is naturally bounded by the $\ell_p$ sensitivity sampling algorithm: indeed, if we a priori assume that $\norm*{\bfS\bfA\bfx}_p^p \geq (1/2)\norm*{\bfA\bfx}_p^p$ for every $\bfx\in\mathbb R^d$, then $\bfsigma_i^p(\bfS\bfA)$ is at most
\begin{equation}\label{eq:SA-sens-bound}
    \sup_{\bfS\bfA\bfx\neq 0}\frac{\abs*{[\bfS\bfA\bfx](i)}^p}{\norm*{\bfS\bfA\bfx}_p^p} \leq 2\sup_{\bfA\bfx\neq 0}\frac{1}{p_i}\frac{\abs*{[\bfA\bfx](i)}^p}{\norm*{\bfA\bfx}_p^p} \leq \frac{2\bfsigma_i^p(\bfA)}{p_i}
\end{equation}
which is at most $2\alpha$ if $p_i \geq \bfsigma_i^p(\bfA) / \alpha$. While this is only an informal argument, this intuition can be formalized, as we discuss later in this section.

Next, we require a bound on the leverage scores $\tau$ of $\bfS\bfA$, which is more challenging, as sensitivity sampling does not directly bound this quantity. To address this problem, we show how to bound the leverage scores of an arbitrary matrix by the $\ell_p$ sensitivities of the matrix. In particular, we show that for $q\geq p$, the largest $\ell_q$ sensitivity bounds the largest $\ell_p$ sensitivity. %

\begin{lemma}[Monotonicity of Max $\ell_p$ Sensitivity]\label{lem:sens-mon}
Let $q\geq p > 0$ and $\bfy\in\mathbb R^n$. Then,
\[
    \frac{\norm{\bfy}_\infty^p}{\norm*{\bfy}_{p}^p} \leq \frac{\norm*{\bfy}_\infty^q}{\norm*{\bfy}_{q}^q}.
\]
\end{lemma}

We also use an ``approximate converse'' of the above result:
\begin{lemma}[Reverse Monotonicity of Max $\ell_p$ Sensitivity]\label{lem:sens-mon-rev}
Let $q\geq p > 0$ and $\bfy\in\mathbb R^n$. Then, 
\[
    \frac{\norm*{\bfy}_\infty^q}{\norm*{\bfy}_q^q} \leq \parens*{\frac{\norm*{\bfy}_\infty^p}{\norm*{\bfy}_p^p}}^{q/p} n^{q/p - 1}.
\]
\end{lemma}

While these lemmas only apply to the max $\ell_p$ sensitivity and are quite loose when the max $\ell_p$ sensitivity can be arbitrary, we use the crucial fact that the $\ell_p$ sensitivities of $\bfS\bfA$ are essentially ``flat'', that is, the maximum $\ell_p$ sensitivity will be within a small factor of the \emph{average} $\ell_p$ sensitivity $\frS^p(\bfA)/n$. Thus, \cref{lem:sens-mon} and \cref{lem:sens-mon-rev} allow us to bound the leverage scores by the $\ell_p$ sensitivities for $p>2$ and $p<2$, respectively. This idea also allows us to prove \cref{thm:total-sens-lb}.

\subsection{Gaussianization Reduction for Sampling Algorithms}
\label{sec:gaussianization-reduction}

In the works of \cite{BLM1989, LT1991}, a version of \cref{lem:rad-bound} tailored to $\ell_p$ Lewis weight sampling is used as a part of a recursive sampling algorithm, where the Rademacher process represents the sampling error of a process that samples each row $i\in[n]$ with probability $1/2$ and scales the result by $2$. Indeed, if $\bfS_{i,i}$ takes the value $0$ or $2^{1/p}$ with probability $1/2$ each and $\norm*{\bfA\bfx}_p^p = 1$, then $(\bfS_{i,i}^p - 1)$ is a Rademacher variable and
\[
    \norm*{\bfS\bfA\bfx}_p^p - 1 = \sum_{i=1}^n (\bfS_{i,i}^p - 1)\abs*{[\bfA\bfx](i)}^p,
\]
and thus \cref{lem:rad-bound} bounds \eqref{eq:sampling-error}. While this only reduces the number of rows by a factor of $2$, this process can be applied recursively for $O(\log n)$ rounds to reduce the number of rows to $\poly(d)$.

However, we instead primarily use \cref{lem:rad-bound} in a reduction based on \cite{CP2015} for an algorithm with \emph{one} round of sampling. In this reduction, we bound the sampling error \eqref{eq:sampling-error} by introducing an independent copy $\bfS'$ of $\bfS$ and estimate
\[
    \Lambda \leq \E\sup_{\norm*{\bfA\bfx}_p = 1}\abs*{\norm*{\bfS\bfA\bfx}_p^p - \norm*{\bfS'\bfA\bfx}_p^p}
\]
Because $\bfS$ and $\bfS'$ are identically distributed, multiplying each coordinate $i\in[n]$ by a Rademacher variable $\bfeps_i\sim\{\pm1\}$ does not change the distribution. Then by applying the triangle inequality, it follows that
\begin{equation}\label{eq:symmetrization-lambda}
    \Lambda \leq 2\E_{\bfeps\sim\{\pm1\}^n}\sup_{\norm*{\bfA\bfx}_p=1}\abs*{\sum_{i=1}^n \bfeps_i \abs*{[\bfS\bfA\bfx](i)}^p},
\end{equation}
which closely resembles \cref{lem:rad-bound}.

At this point, for each fixing of $\bfS$, we wish to apply \cref{lem:rad-bound} with $\bfA$ replaced by $\bfS\bfA$, with sensitivities bounded using the idea of \eqref{eq:SA-sens-bound}. However, we cannot a priori assume that $\norm*{\bfS\bfA\bfx}_p^p \geq (1/2)\norm*{\bfA\bfx}_p^p$. To fix this, the idea of \cite{CP2015} is to introduce an \emph{auxiliary subspace embedding} $\bfS'$ such that $\bfS'\bfA$ also has $\ell_p$ sensitivities bounded by $\alpha$, and \emph{does} satisfy $\norm*{\bfS'\bfA\bfx}_p^p = \Theta(1)\norm*{\bfA\bfx}_p^p$ for every $\bfx\in\mathbb R^d$. Then, we can apply the result of \cref{lem:rad-bound} to the concatenated matrix
\[
    \bfA' \coloneqq \begin{pmatrix}\bfS\bfA \\ \bfS'\bfA\end{pmatrix}.
\]
It can then be shown that the quantity $E$ bounded in \cref{lem:rad-bound} indeed bounds the quantity in \eqref{eq:symmetrization-lambda}, and furthermore, we can fix the argument of \eqref{eq:SA-sens-bound} by bounding
\[
    \sup_{\bfA'\bfx\neq 0}\frac{\abs*{[\bfS\bfA\bfx](i)}^p}{\norm*{\bfA'\bfx}_p^p} \leq \sup_{\bfA\bfx\neq 0}\frac{1}{p_i}\frac{\abs*{[\bfA\bfx](i)}^p}{\norm*{\bfS'\bfA\bfx}_p^p} \leq \frac{O(\bfsigma_i^p(\bfA))}{p_i}.
\]
Importantly, we only need to use the existence of $\bfS'$ for the analysis and thus $\bfS'$ can be constructed in any way, rather than just by sensitivity sampling. 

\subsection{Construction of Auxiliary \texorpdfstring{$\ell_p$}{lp} Subspace Embeddings}

We now delve into further detail about the other key piece of our analysis, which is the construction of auxiliary $\ell_p$ subspace embeddings $\bfS'$ which are compatible with the reduction argument in \cref{sec:gaussianization-reduction}. 

\subsubsection{Sensitivity Sampling, \texorpdfstring{$p<2$}{p<2}}

For our sensitivity sampling result for $p<2$ in \cref{thm:main-informal}, we aim to sample roughly $m = \eps^{-2}\frS^p(\bfA)^{2/p}$ rows. We first briefly sketch the intuition behind this bound. If we sample each row $i\in[n]$ with probability roughly $\bfsigma_i^p(\bfA)/\alpha$ for an oversampling parameter $\alpha>0$, then we expect the $\ell_p$ sensitivities of $\bfS\bfA$ to be bounded by $\alpha$ according to the informal reasoning of \eqref{eq:SA-sens-bound}. We then wish to bound the sampling error using \cref{lem:rad-bound}, which, combined with the bound on the leverage scores using the sensitivities in \cref{lem:sens-mon-rev}, gives a bound of
\[
    \tau \leq \alpha^{2/p}m^{2/p-1} = \alpha\frS^p(\bfA))^{2/p-1}
\]
by using that $m$ is $\frS^p(\bfA) / \alpha$ in expectation. Since we want this to be at most $\eps^2$, we set $\alpha = \eps^2 / \frS^p(\bfA)^{2/p-1}$, which gives a bound of $m = \eps^{-2}\frS^p(\bfA)^{2/p}$ as claimed.

Now coming back to the formal argument using the auxiliary $\ell_p$ subspace embedding $\bfS'$, we wish to construct a subspace embedding $\bfS'$ that preserves $\ell_p$ norms of $\bfA\bfx$ up to a constant factor, but also has $\ell_p$ sensitivities at most $\alpha$, since this is the bound that $\bfS\bfA$ will satisfy. To construct this $\bfS'$, we proceed by first using Lewis weight sampling (\cref{thm:lewis-weight-sampling}) to construct a $\Theta(1)$-approximate $\ell_p$ subspace sampling with $\tilde O(d)$ rows. Importantly, we show that this sampling procedure can only increase the total $\ell_p$ sensitivity by a constant factor (see \cref{lem:sampling-preserve-sens}). We then apply a flattening procedure (see \cref{lem:sens-flat}), which yields an $\ell_p$ isometry with at most $\frS^p(\bfA) / \alpha$ rows such that the sensitivity of each row is bounded by $\alpha$. Furthermore, because the number of rows $\frS^p(\bfA)/\alpha$ of this auxiliary subspace embedding $\bfS'$ is at most the bound $m$ that we seek, we can apply the same reasoning as before using \cref{lem:sens-mon-rev} to bound the leverage scores of $\bfS'\bfA$ to recover the same bound as that for $\bfS\bfA$. Thus, we obtain our desired construction of the auxiliary $\ell_p$ subspace embedding, allowing the reduction argument as described in \cref{sec:gaussianization-reduction} to go through.

\subsubsection{Sensitivity Sampling, \texorpdfstring{$p>2$}{p>2}}

For our result on sensitivity sampling when $p>2$, the intuition and reasoning roughly follows the case of $p<2$. However, we must be more careful with the construction of the auxiliary $\ell_p$ subspace embedding, since we cannot use Lewis weight sampling to construct it; this is because the sample complexity of $\tilde O(d^{p/2})$ for Lewis weight sampling is larger than our sample complexity bound of $m = \eps^{-2}\frS^p(\bfA)^{2-2/p}$ that we aim for. Thus, in order to obtain our auxiliary $\ell_p$ subspace embedding, we must essentially achieve the same results that we claim in \cref{thm:main-informal}, using an \emph{alternate construction} that does not directly use one-shot sensitivity sampling. For this, we instead revisit the \emph{recursive sampling} process mentioned earlier in \cref{sec:gaussianization-reduction}, which originates from the functional analysis literature. Here, we first flatten our input matrix using \cref{lem:sens-flat} to a matrix with at most $(4/3) n$ rows and $\ell_p$ sensitivities at most $O(\frS^p(\bfA)/n)$, and then use \cref{lem:rad-bound} directly as the sampling error bound for a uniform sampling algorithm which samples each row with probability $1/2$. Note then that overall, we retain $(4/3) (1/2) n = (2/3)n$ rows altogether in expectation after this result, and furthermore, we have a sampling error of at most $\eps$. Then, because we sample only a constant fraction of rows after each application of the procedure, it can be shown that recursively applying this result for $O(\log n)$ iterations accumulates a total sampling error of $O(\eps\log n)$, while reducing the number of rows down to $\eps^{-2}\frS^p(\bfA)^{2-2/p}$, which matches the number of rows that we claim for sensitivity sampling in \cref{thm:main-informal}. By rescaling $\eps$ by an $O(\log n)$ factor, we can obtain a total error of $\eps$ while only losing polylogarithmic factors in $n$. This is carried out in \cref{lem:recursive-sens-sampling}. Thus, we may again carry out the reduction argument as described in \cref{sec:gaussianization-reduction}. 

\subsubsection{Root Leverage Score Sampling, \texorpdfstring{$p<2$}{p<2}}
\label{sec:root-lev-aux}

For our root leverage score sampling theorem \cref{thm:informal-root-lev}, we take a conceptually different approach for obtaining the leverage score bound required in \cref{lem:rad-bound}. For sensitivity sampling, our idea was to bound the leverage scores by relating them to the $\ell_p$ sensitivities, which were controlled by the sensitivity sampling process. For root leverage score sampling, however, the idea is that by sampling by the root leverage scores, we directly control the \emph{leverage scores} of the sampled matrix $\bfS\bfA$, rather than the sensitivity scores. Indeed, one can show that if the sampling probabilities $p_i$ satisfy $p_i \geq \bftau_i(\bfA)^{p/2}/\alpha$, and if $\norm*{\bfS\bfA\bfx}_2^2 \geq (1/2) \norm*{\bfA\bfx}_2^2$ for every $\bfx\in\mathbb R^d$, then the leverage scores of $\bfS\bfA$ are bounded by
\begin{equation}\label{eq:root-lev-intuition}
    \sup_{\bfS\bfA\bfx\neq 0}\frac{\abs*{[\bfS\bfA\bfx](i)}^2}{\norm*{\bfS\bfA\bfx}_2^2} \leq 2\sup_{\bfA\bfx\neq 0}\frac{1}{p_i^{2/p}}\frac{\abs*{[\bfA\bfx](i)}^2}{\norm*{\bfA\bfx}_2^2} \leq 2\frac{\bftau_i(\bfA)}{p_i^{2/p}}
\end{equation}
which is at most $2\alpha^{2/p}$. However, it is not clear that $\norm*{\bfS\bfA\bfx}_2^2 \geq (1/2) \norm*{\bfA\bfx}_2^2$ should hold at all, since $\bfS$ is only a subspace embedding for $\ell_p$ and not $\ell_2$. For this intuition to go through, we require an auxiliary $\ell_p$ subspace embedding that preserves both $\ell_p$ norms and $\ell_2$ norms, with flat $\ell_2$ leverage scores.

To make this argument work, we crucially use the fact that the bound in \cref{lem:rad-bound} does not depend polynomially on the number of rows of the input matrix. Thus, we can afford to construct $\bfS'\bfA$ to have many rows, as long as its leverage scores are controlled. Thus, in \cref{lem:flatten-all}, we take the approach of constructing an $\ell_p$ isometry of $\bfA$ by splitting every row $\bfa_i$ into $k$ copies of $\bfa_i / k^{1/p}$, where we will set $k = 1 / \alpha$. Note then that after splitting, every row contains at most a $1/k$ fraction of the $\ell_2$ mass, so the $\ell_2$ leverage scores are all at most $\alpha$ (in fact, all $\ell_q$ sensitivities for any $q$ are at most $\alpha$). Furthermore, crucially, the $\ell_2$ norm will be preserved up to a factor of $\alpha^{1/p-1/q}$. Then altogether, we can fix \eqref{eq:root-lev-intuition} and instead bound
\[
    \sup_{\bfS\bfA\bfx\neq 0}\frac{\abs*{[\bfS\bfA\bfx](i)}^2}{\norm*{\bfA'\bfx}_2^2} \leq 2\sup_{\bfA\bfx\neq 0}\frac{1}{p_i^{2/p}}\frac{\abs*{[\bfA\bfx](i)}^2}{\alpha^{2/p-1}\norm*{\bfA\bfx}_2^2}
\]
which is at most $2\alpha$, where $\bfA'$ is the concatenation of $\bfS\bfA$ and $\bfS'\bfA$. These ideas lead to our \cref{thm:informal-root-lev}.

\subsubsection{Leverage Score + Sensitivity Sampling, \texorpdfstring{$p>2$}{p>2}}
\label{sec:lev-sens-aux}

Finally, for our \cref{thm:sens-lev-informal}, we take essentially the same approach as the recursive form of our sensitivity sampling result for $p>2$ in \cref{thm:main-informal}, except that we improve our flattening approach by flattening \emph{leverage scores} as well, which is formalized in \cref{lem:flat-sens-lev}. Note that when we only flatten $\ell_p$ sensitivities using \cref{lem:sens-flat}, then the resulting bound on the leverage scores is just roughly the average $\ell_p$ sensitivity, which is $\frS^p(\bfA)/n$. We show that by also flattening the leverage scores, we can improve this bound to
\[
    \parens*{\frac{d}{n}}^{2/p}\parens*{\frac{\frS^p(\bfA)}n}^{1-2/p}.
\]
Because $\frS^p(\bfA)$ is always at least $d/2$ for $p>2$ due to our \cref{thm:total-sens-lb}, this is always better than the previous bound of $\frS^p(\bfA)/n$, which ultimately leads to our improved sampling algorithm \cref{thm:sens-lev-informal}.

We note that we do not obtain a corresponding one-shot sampling algorithm, where the main difficulty is in constructing an auxiliary $\ell_p$ subspace embedding that has few rows, flat leverage scores, does not increase $\ell_p$ norms, and does not decrease $\ell_2$ norms. Nonetheless, the recursive sampling procedure still leads to an efficient algorithm.

\section{Conclusion and Future Directions}

Our work introduces a new analysis for sensitivity sampling for $\ell_p$ subspace embeddings, which breaks a previous general sampling barrier of $\tilde O(\eps^{-2}\frS^p(\bfA)d)$ samples via a simple union bound argument, to obtain an improved bound of $\tilde O(\eps^{-2}\frS^p(\bfA)^{2/p})$ samples for $p<2$ and $\tilde O(\eps^{-2}\frS^p(\bfA)^{2-2/p})$ samples for $p>2$. We also present other novel results for sampling algorithms for $\ell_p$ subspace embeddings based on our techniques, showing that the popular root leverage score sampling algorithm yields a bound of $\tilde O(\eps^{-4/p}d)$ for $p<2$, as well as an improved $\tilde O(\eps^{-2}d^{2/p}\frS^p(\bfA)^{2-4/p})$ bound for $p>2$ using a recursive sampling algorithm that combines $\ell_p$ sensitivity flattening with leverage score flattening. Our improved analyses of sensitivity sampling as well as our novel leverage score and sensitivity flattening algorithm give the best known sampling guarantees for a number of structured regression problems with small arbitrary noise. 

We conclude with several open questions. Perhaps the most natural is to completely resolve Question~\ref{q:main} by characterizing the sample complexity of sensitivity sampling. We conjecture that a sample complexity of $\tilde O(\eps^{-2}(\frS^p(\bfA) + d))$ is possible for $\ell_p$ subspace embeddings, and perhaps for more broad settings where sensitivity sampling applies as well. Furthermore, for $p>2$, we believe it is of interest to obtain this bound even without the use of sensitivity sampling via other methods. Finally, we raise the question of obtaining sampling algorithms for subspace embeddings for the Huber loss with nearly optimal sample complexity, for which our results may be useful. 

\section*{Acknowledgements}

We thank the anonymous reviewers for useful feedback on improving the presentation of this work. David P.\ Woodruff and Taisuke Yasuda were supported by a Simons Investigator Award.

%% file: appendix.tex
\section{Properties of \texorpdfstring{$\ell_p$}{lp} Sensitivities}
\label{sec:lp-sens-proof}

\subsection{Monotonicity of Max \texorpdfstring{$\ell_p$}{lp} Sensitivity}

We first provide proofs of \cref{lem:sens-mon} and \cref{lem:sens-mon-rev}. The results are similar to results used in \cite{BDMMUWZ2020, MMWY2022}. In particular, it generalizes Lemma 4.6 of \cite{BDMMUWZ2020} and is a simplification of a specific instance of Lemma C.3 of \cite{MMWY2022}.

\begin{proof}[Proof of \cref{lem:sens-mon}]
We have that
\[
    \norm*{\bfy}_{q}^q = \sum_{i=1}^n \abs{\bfy(i)}^q \leq \norm*{\bfy}_\infty^{q-p}\sum_{i=1}^n \abs{\bfy(i)}^{p} = \norm*{\bfy}_\infty^{q-p} \norm*{\bfy}_{p}^p,
\]
so
\[
    \frac{\norm{\bfy}_\infty^p}{\norm*{\bfy}_{p}^p} \leq \frac{\norm{\bfy}_\infty^p}{\norm*{\bfy}_{q}^q / \norm*{\bfy}_\infty^{q-p}} = \frac{\norm*{\bfy}_\infty^q}{\norm*{\bfy}_{q}^q}.
\]
\end{proof}

\begin{proof}[Proof of \cref{lem:sens-mon-rev}]
Since $\norm*{\bfy}_p \leq \norm*{\bfy}_q n^{1/p-1/q}$, we have that
\[
    \frac{\norm*{\bfy}_\infty^q}{\norm*{\bfy}_q^q} \leq \frac{\abs{\bfy(i)}^q}{\norm*{\bfy}_p^q \cdot n^{1-q/p}} \leq \frac{\norm{\bfy}_\infty^q}{\norm*{\bfy}_p^q \cdot n^{1-q/p}} = \parens*{\frac{\norm*{\bfy}_\infty^p}{\norm*{\bfy}_p^p}}^{q/p} n^{q/p - 1}.
\]
\end{proof}

\subsection{Total Sensitivity}

We now derive bounds on the total $\ell_p$ sensitivity. 

\subsubsection{Sampling Preserves Total Sensitivity}

\begin{lemma}[Sampling Preserves Total Sensitivity]
\label{lem:sampling-preserve-sens}
Let $\bfA\in\mathbb R^{n\times d}$ and $1\leq p < \infty$. Let $\bfS$ be a random $\ell_p$ sampling matrix such that with probability at least $3/4$,
\[
    \norm*{\bfS\bfA\bfx}_p = (1\pm 1/2)\norm*{\bfA\bfx}_p
\]
simultaneously for every $\bfx\in\mathbb R^d$. Then, with probability at least $1/2$,
\[
    \Pr\braces*{\frS^p(\bfS\bfA) \leq 8\frS^p(\bfA)} \geq \frac12.
\]
\end{lemma}
\begin{proof}
We have that
\[
    \frS^p(\bfS\bfA) = \sum_{i=1}^n \sup_{\bfS\bfA\bfx\neq 0}\frac{\abs*{[\bfS\bfA\bfx](i)}^p}{\norm*{\bfS\bfA\bfx}_p^p} = \sum_{i=1}^n \bfS_{i,i}^p \sup_{\bfS\bfA\bfx\neq 0}\frac{\abs*{[\bfA\bfx](i)}^p}{\norm*{\bfA\bfx}_p^p} \frac{\norm*{\bfA\bfx}_p^p}{\norm*{\bfS\bfA\bfx}_p^p} \leq \sum_{i=1}^n \bfS_{i,i}^p\bfsigma_i^p(\bfA)\sup_{\bfS\bfA\bfx\neq 0}\frac{\norm*{\bfA\bfx}_p^p}{\norm*{\bfS\bfA\bfx}_p^p}.
\]
We are guaranteed that
\[
    \Pr\braces*{\sup_{\bfS\bfA\bfx\neq 0}\frac{\norm*{\bfA\bfx}_p^p}{\norm*{\bfS\bfA\bfx}_p^p} \leq 2} \geq \frac34.
\]
On the other hand, we have that
\[
    \E \bracks*{\sum_{i=1}^n\bfS_{i,i}^p\bfsigma_i^p(\bfA)} = \sum_{i=1}^n \E[\bfS_{i,i}^p]\bfsigma_i^p(\bfA) = \frS^p(\bfA)
\]
so by Markov's inequality,
\[
    \Pr\braces*{\sum_{i=1}^n\bfS_{i,i}^p\bfsigma_i^p(\bfA) \leq 4\frS^p(\bfA)} \geq \frac34.
\]
By a union bound,
\[
    \Pr\braces*{\frS^p(\bfS\bfA) \leq 8\frS^p(\bfA)} \geq \frac12.
\]
\end{proof}

We also prove a high probability and high accuracy version of \cref{lem:sampling-preserve-sens}.

\begin{lemma}[Sensitivity Sampling Preserves Total Sensitivity: High Probability and Accuracy]
\label{lem:sampling-preserve-sens-hp}
Let $\bfA\in\mathbb R^{n\times d}$ and $1\leq p < \infty$. Let $0<\eps,\delta<1$. Let $\bfS$ be a random $\ell_p$ sampling matrix such that with probability at least $1-\delta$,
\[
    \norm*{\bfS\bfA\bfx}_p = (1\pm \eps)\norm*{\bfA\bfx}_p
\]
simultaneously for every $\bfx\in\mathbb R^d$. Furthermore, suppose that
\[
    \frac{\bfsigma_i}{q_i} \leq M \coloneqq \frac{\eps^2\frS^p(\bfA)}{3\log\frac2\delta}
\]
for every $i\in[n]$. Then, with probability at least $1-2\delta$,
\[
    \Pr\braces*{\frS^p(\bfS\bfA) = (1\pm O(\eps))\frS^p(\bfA)} \geq 1-2\delta.
\]
\end{lemma}
\begin{proof}
The proof follows \cref{lem:sampling-preserve-sens}. Just as in \cref{lem:sampling-preserve-sens}, we have that
\[
    \frS^p(\bfS\bfA) \leq \sum_{i=1}^n \bfS_{i,i}^p\bfsigma_i^p(\bfA)\sup_{\bfS\bfA\bfx\neq 0}\frac{\norm*{\bfA\bfx}_p^p}{\norm*{\bfS\bfA\bfx}_p^p}.
\]
Similarly,
\[
    \frS^p(\bfS\bfA) \geq \sum_{i=1}^n \bfS_{i,i}^p\bfsigma_i^p(\bfA)\inf_{\bfS\bfA\bfx\neq 0}\frac{\norm*{\bfA\bfx}_p^p}{\norm*{\bfS\bfA\bfx}_p^p}.
\]
Furthermore, since $\bfsigma_i / q_i \leq M$, $\bfS_{i,i}^p\bfsigma_i^p(\bfA) / M$ is a random variable bounded by $1$, with
\[
    \E\bracks*{\sum_{i=1}^n \frac{\bfS_{i,i}^p\bfsigma_i^p(\bfA)}{M}} = \frac{\frS^p(\bfA)}{M} \geq \frac3{\eps^2}\log\frac2\delta.
\]
Thus by Chernoff bounds, we have that
\[
    \Pr\braces*{\sum_{i=1}^n \bfS_{i,i}^p\bfsigma_i^p(\bfA) = (1\pm\eps)\frS^p(\bfA)} \geq 1 - \delta.
\]
We conclude by a union bound as in \cref{lem:sampling-preserve-sens}.
\end{proof}

\subsubsection{Total Sensitivity Lower Bounds}

We start with the classical result that the total $\ell_2$ sensitivity is exactly $d$:

\begin{lemma}
\label{lem:lev-score-total-sens}
Let $\bfA\in\mathbb R^{n\times d}$ and let $\bfU\in\mathbb R^{n\times d}$ be an orthnormal basis for the column space of $\bfA$. Then,
\[
    \bftau_i(\bfA) = \norm*{\bfe_i^\top\bfU}_2^2
\]
and
\[
    \sum_{i=1}^n \bftau_i(\bfA) = \norm*{\bfU}_F^2 = d.
\]
\end{lemma}
\begin{proof}
We have that
\[
    \bftau_i(\bfA) = \sup_{\bfx\in\mathbb R^d, \bfA\bfx\neq 0}\frac{\abs*{[\bfA\bfx](i)}^2}{\norm*{\bfA\bfx}_2^2} = \sup_{\bfx\in\mathbb R^d, \bfU\bfx\neq 0}\frac{\abs*{[\bfU\bfx](i)}^2}{\norm*{\bfU\bfx}_2^2} = \sup_{\bfx\in\mathbb R^d, \bfx\neq 0}\frac{\abs*{[\bfU\bfx](i)}^2}{\norm*{\bfx}_2^2} = \norm*{\bfe_i^\top\bfU}_2^2.
\]
\end{proof}

We now use \cref{lem:lev-score-total-sens} together with \cref{lem:sens-mon} and \cref{lem:sens-mon-rev} to derive lower bounds on $\frS^p(\bfA)$.

By using a simple argument based on ``splitting rows'' (see, e.g., \cite{LT1991, CP2015, CD2021, MMWY2022}), it is possible to assume without loss of generality that the maximum $\ell_p$ sensitivity is related to the average $\ell_p$ sensitivity, up to a factor of $2$:

\begin{lemma}[$\ell_p$ Sensitivity Flattening]\label{lem:sens-flat}
Let $\bfA\in\mathbb R^{n\times d}$ and $1 \leq p < \infty$. Let $C\geq 1$. Then, there exists a $\bfA'\in\mathbb R^{m\times d}$ for $m = (1+1/C)n$ such that $\norm*{\bfA\bfx}_p = \norm*{\bfA'\bfx}_p$ for every $\bfx\in\mathbb R^d$, $\frS^p(\bfA) = \frS^p(\bfA')$, and $\bfsigma_{i'}^p(\bfA') \leq C\frS^p(\bfA) / n$ for every $i'\in [m]$.
\end{lemma}
\begin{proof}[Proof of \cref{lem:sens-flat}]
Suppose that for any row $\bfa_i\in\mathbb R^d$ of $\bfA$ for $i\in[n]$ with $\bfsigma_i^p(\bfA) \geq C\frS^p(\bfA) / n$, we replace the row with $k\coloneqq \ceil*{\bfsigma_i^p(\bfA) / (C\frS^p(\bfA) / n)}$ copies of $\bfa_i / k^{1/p}$ to form a new matrix $\bfA'$. Then, we add at most
\[
    \sum_{i : \bfsigma_i^p(\bfA) \geq \frS^p(\bfA) / n}\ceil*{\frac{\bfsigma_i^p(\bfA)}{C\frS^p(\bfA)/n}} - 1 \leq \sum_{i : \bfsigma_i^p(\bfA) \geq \frS^p(\bfA) / n}\frac{\bfsigma_i^p(\bfA)}{C\frS^p(\bfA)/n} = \frac{\frS^p(\bfA)}{C\frS^p(\bfA)/n} = \frac{n}{C}
\]
rows. Furthermore, we clearly have that $\norm*{\bfA\bfx}_p = \norm*{\bfA'\bfx}_p$ for every $\bfx\in\mathbb R^d$, and also for any row $i'\in[m]$ that comes from row $i\in[n]$ in the original matrix,
\[
    \frac{\abs*{[\bfA'\bfx](i')}^p}{\norm*{\bfA'\bfx}_p^p} \leq \frac{C\frS^p(\bfA)/n}{\bfsigma_i^p(\bfA)}\frac{\abs*{[\bfA\bfx](i)}^p}{\norm*{\bfA\bfx}_p^p} \leq \frac{C\frS^p(\bfA)}{n}.
\]
Finally, it is also clear that the sum of the sensitivities is also preserved, since the sum of the sensitivities of the $k$ copies of each row $i\in[n]$ in the original matrix is $\bfsigma_i^p(\bfA)$.
\end{proof}

We can now prove \cref{thm:total-sens-lb}:

\begin{proof}[Proof of \cref{thm:total-sens-lb}]
Let $\bfA'\in\mathbb R^{2n\times d}$ be the matrix given by \cref{lem:sens-flat} applied with $C = 1$. Then for $p > 2$, we have by \cref{lem:sens-mon} that
\[
    \frac{d}{n} \leq \max_{i=1}^n \bfsigma_i^2(\bfA') \leq \max_{i=1}^n \bfsigma_i^p(\bfA') \leq \frac{2\frS^p(\bfA)}{n}
\]
and for $p < 2$, we have by \cref{lem:sens-mon-rev} that
\[
    \frac{d}{n} \leq \max_{i=1}^n \bfsigma_i^2(\bfA') \leq \parens*{\max_{i=1}^n \bfsigma_i^p(\bfA')}^{2/p} n^{2/p-1} \leq \parens*{\frac{2\frS^p(\bfA)}{n}}^{2/p} n^{2/p-1} = \frac{2^{2/p}\frS^p(\bfA)^{2/p}}{n}
\]
which yield the claimed results.
\end{proof}

\subsubsection{Random Matrices Have Small Total Sensitivity}

We show that the above lower bounds can be tight, up to logarithmic factors. We will use Dvoretzky's theorem, which can be found in, e.g., Fact 15 in \cite{SW2018}:

\begin{theorem}[Dvoretzky's Theorem]\label{thm:dvoretzky}
Let $1\leq p < 2$. Let $n = (d/\eps)^{O(p)}$ be sufficiently large, and let $\bfA$ be a suitably scaled random $n\times d$ Gaussian matrix. Then, with probability at least $99/100$, we have for every $\bfx\in\mathbb R^d$ that
\[
    \norm*{\bfA\bfx}_p = (1\pm\eps)\norm*{\bfx}_2.
\]
\end{theorem}

This gives a proof of \cref{thm:small-total-sens}:

\begin{proof}[Proof of \cref{thm:small-total-sens}]
By applying \cref{thm:dvoretzky} with $\eps = 1/2$, we have that $\norm*{\bfA\bfx}_p^p = \Theta(n)\norm*{\bfx}_2^p$ with probability at least $99/100$. Note also that
\[
    \max_{i=1}^n \norm*{\bfe_i^\top\bfA}_2 \leq O(\sqrt{d\log n}) = O(\sqrt{d\log d})
\]
which also happens with probability at least $99/100$. By a union bound, both events happen with probability at least $98/100$.

Now for any $\bfx\in\mathbb R^d$ with unit $\ell_2$ norm, we have that
\[
    \abs*{[\bfA\bfx](i)}^p \leq \norm*{\bfe_i^\top\bfA}_2^p\cdot\norm*{\bfx}_2^p = \norm*{\bfe_i^\top\bfA}_2^p \leq O(d\log d)^{p/2}.
\]
Thus,
\[
    \frS^p(\bfA) \leq n\cdot \max_{i=1}^n \sup_{\norm*{\bfx}_2 = 1}\frac{\abs*{[\bfA\bfx](i)}^p}{\norm*{\bfA\bfx}_p^p} \leq n\cdot \max_{i=1}^n \sup_{\norm*{\bfx}_2 = 1}\frac{O(d\log d)^{p/2}}{\Theta(n)} = O(d\log d)^{p/2}.
\]
\end{proof}

\subsection{Structured Matrices with Small Sensitivity, \texorpdfstring{$p>2$}{p>2}}
\label{sec:structured-mat-sens}

\begin{proof}[Proof of \cref{lem:sens-low-rank-sparse}]
Let $r$ be an integer such that $2^r \leq p < 2^{r+1}$. Then, for each $i\in[n]$, we may write
\[
    \bfa_i = \bfk_i + \bfs_i = \sum_{j=1}^k \alpha_{i,j}\bfv_j + \sum_{j=1}^s \beta_{i,j}\bfe_{i_j}
\]
where $\bfv_j\in\mathbb R^d$ for $j\in[k]$. Then, the tensor product $\bfa_i^{\otimes 2^r}$ of $\bfa_i$ with itself $2^r$ times can be written as a linear combination of tensor products $\bfy_1 \otimes \dots \otimes \bfy_{2^r}$, where each $\bfy_q$ for $q\in[2^r]$ is one of $\{\bfv_1, \bfv_2, \dots, \bfv_k, \bfe_{i_1}, \bfe_{i_2}, \dots, \bfe_{i_s}\}$. Thus, $\bfa_i^{\otimes r}$ lies in the span of at most $(k+s)^{2^r}$ vectors, for a fixed choice of $\bfe_{i_1}, \bfe_{i_2}, \dots, \bfe_{i_s}$. Since there are at most $d^s$ possible choices of the sparsity pattern, every $\bfa_i^{\otimes 2^r}$ for $i\in[n]$ lies in the span of at most $d' \coloneqq d^s (k+s)^{2^r}$ vectors. That is, if $\bfA^{\otimes 2^r}$ is the Khatri-Rao $2^r$th power of $\bfA$, then $\bfA^{\otimes 2^r}$ is a rank $d'$ matrix. Then, we have that
\[
    \abs*{[\bfA\bfx](i)}^p = (\abs*{[\bfA\bfx](i)}^{2^r})^{p/2^r} = (\angle*{\bfa_i, \bfx}^{2^r})^{p/2^r} = \abs*{\angle*{\bfa_i^{\otimes 2^r}, \bfx^{\otimes 2^r}}}^{p/2^r}
\]
so
\[
    \sup_{\bfA\bfx\neq 0}\frac{\abs*{[\bfA\bfx](i)}^p}{\norm*{\bfA\bfx}_p^p} = \sup_{\bfA\bfx\neq 0}\frac{\abs*{[\bfA^{\otimes 2^r}\bfx^{\otimes 2^r}](i)}^{p/2^r}}{\norm*{\bfA^{\otimes 2^r}\bfx^{\otimes 2^r}}_{p/2^r}^{p/2^r}} \leq \sup_{\bfA^{\otimes 2^r}\bfx\neq 0}\frac{\abs*{[\bfA^{\otimes 2^r}\bfx](i)}^{p/2^r}}{\norm*{\bfA^{\otimes 2^r}\bfx}_{p/2^r}^{p/2^r}} 
\]
that is, the $\ell_{p/2^r}$ sensitivities of $\bfA^{\otimes 2^r}$ upper bound the $\ell_p$ sensitivities of $\bfA$. Since $p/2^r \leq 2$, the total $\ell_{p/2^r}$ sensitivity of $\bfA^{\otimes 2^r}$ is bounded by its rank, which is $d'$.
\end{proof}

\begin{proof}[Proof of Lemma \ref{lem:poly-feature-map-sens}]
Let $r$ be an integer such that $2^r\leq p < 2^{r+1}$. Fix some $\bfx\in\mathbb R^{k(q+1)}$. Now consider the vector $\angle*{\bfa, \bfx}$, where $\bfa$ is a $k(q+1)$-dimensional vector of monomials of degree $0$ through $q$ of the indeterminate variables $a_1, a_2, \dots, a_k$, that is,
\[
    \bfa = (1, a_1, a_1^2, \dots, a_1^q, \quad 1, a_2, a_2^2, \dots, a_2^q, \quad \dots, \quad 1, a_k, a_k^2, \dots, a_k^q).
\]
Then, $\angle*{\bfa,\bfx}$ is a degree $q$ polynomial in the indeterminates $a_1, a_2, \dots, a_k$ with coefficients specified by $\bfx$, so $\angle*{\bfa,\bfx}^{2^r}$ is a polynomial in the indeterminates $a_1, a_2, \dots, a_k$, such that every monomial term is at most degree $2^r q$ in each variable. Note that there are at most $k$ variables, so there can be at most $(2^r q + 1)^k$ possible monomials, by choosing the degree of each of the monomials. Let $\bfx'$ denote the coefficients of this polynomial in the monomial basis, for a given set of original coefficients $\bfx$.

Now consider the matrix $V^q(\bfA)$. Then, for a fixed $\bfx\in\mathbb R^{k(q+1)}$, $[V^q(\bfA)\bfx](i)^{2^r}$ is the evaluation of $\angle*{\bfa,\bfx}^{2^r}$ at the $i$th row $\bfa_i$ of $\bfA$ for the indeterminates $a_1, a_2, \dots, a_k$, so it can be written as the linear combination of at most $(2^r q+1)^k$ monomials evaluated at $\bfa_i$, with coefficients $\bfx'$. Thus, $[V^q(\bfA)\bfx](i)^{2^r} = \bfA'\bfx'$ for some $\bfA'$ with rank at most $(2^r q+1)^k$. 

Finally, note that
\[
    \abs*{[V^q(\bfA)\bfx](i)}^p = (\abs*{[V^q(\bfA)\bfx](i)}^{2^r})^{p/2^r} = \abs*{[\bfA'\bfx'](i)}^{p/2^r}.
\]
Thus, the total $\ell_p$ sensitivity of $V^q(\bfA)$ is bounded by the total $\ell_{p/2^r}$ sensitivity of $\bfA'$, which is at most $(2^r q + 1)^k \leq (pq+1)^k$. 
\end{proof}

\subsection{Total \texorpdfstring{$\ell_p$}{lp} Sensitivity Under Perturbations}
\label{sec:sens-noisy-matrix}

\begin{proof}[Proof of \cref{lem:sens-noisy-matrix}]
For any $\bfx\in\mathbb R^d$, we have that
\begin{align*}
    \norm*{(\bfA+\bfE)\bfx}_p &= \norm*{\bfA\bfx}_p \pm \norm*{\bfE\bfx}_p \\
    &= \norm*{\bfA\bfx}_p \pm \sqrt n\norm*{\bfE\bfx}_2 \\
    &= \norm*{\bfA\bfx}_p \pm \frac{\sigma_{\min}}{\sqrt n}\norm*{\bfx}_2 \\
    &= \norm*{\bfA\bfx}_p \pm \frac{\sigma_{\min}}{\sqrt n}\frac{1}{\sigma_{\min}}\norm*{\bfA\bfx}_2 \\
    &= \norm*{\bfA\bfx}_p \pm \frac12\norm*{\bfA\bfx}_p \\
    &= (1\pm 1/2) \norm*{\bfA\bfx}_p
\end{align*}
so
\[
    \frac{\abs*{[(\bfA+\bfE)\bfx](i)}^p}{\norm*{(\bfA+\bfE)\bfx}_p^p} \leq 2^{p-1}\frac{\abs*{[\bfA\bfx](i)}^p}{\norm*{(\bfA+\bfE)\bfx}_p^p} + 2^{p-1}\frac{\abs*{[\bfE\bfx](i)}^p}{\norm*{(\bfA+\bfE)\bfx}_p^p} \leq 2^{p}\frac{\abs*{[\bfA\bfx](i)}^p}{\norm*{\bfA\bfx}_p^p} + 2^{p}\frac{\abs*{[\bfE\bfx](i)}^p}{\norm*{\bfA\bfx}_p^p}.
\]
The first term is clearly bounded by $2^p\bfsigma_i^p(\bfA)$ for any $\bfx$. On the other hand, the second term is bounded by
\[
    2^{p}\frac{\abs*{[\bfE\bfx](i)}^p}{\norm*{\bfA\bfx}_p^p} \leq 2^{p}\frac{\norm*{\bfE\bfx}_p^p}{\norm*{\bfA\bfx}_p^p} \leq 2^p  n^{p/2}\frac{\norm*{\bfE\bfx}_2^p}{\norm*{\bfA\bfx}_p^p} \leq 2^p  \frac{\sigma_{\min}^p}{n^{p/2+1}}\frac{\norm*{\bfx}_2^p}{\norm*{\bfA\bfx}_p^p} \leq 2^p  \frac{1}{n^{p/2+1}}\frac{\norm*{\bfA\bfx}_2^p}{\norm*{\bfA\bfx}_p^p} \leq 2^p \frac{\norm*{\bfA\bfx}_p^p}{n\norm*{\bfA\bfx}_p^p} = \frac{2^p}{n}.
\]
Thus, the total sensitivity is bounded by
\[
    2^p\sum_{i=1}^n \bfsigma_i^p(\bfA) + \frac1n = 2^p(\frS^p(\bfA) + 1).
\]
\end{proof}

\section{Entropy Estimates}
\label{sec:entropy-est}

In this section, we collect our results on estimates on various metric entropies, which are needed for our chaining arguments. Our results here are based on similar results given by \cite{BLM1989}. However, we modify their arguments to only depend on leverage scores and $\ell_p$ sensitivities, rather than using Lewis weights \cite{Lew1978, BLM1989}.

\subsection{Preliminaries}

We first recall general definitions from convex geometry that are relevant to this section.

\begin{definition}[$d_X$-balls]\label{def:dx-ball}
Let $d_X$ be a metric on $\mathbb R^d$. Then, for $\bfx\in\mathbb R^d$ and $t\geq 0$, we define the $d_X$-ball of radius $t$ $B_d(\bfx,t)$ to be
\[
    B_X(\bfx,t)\coloneqq \braces*{\bfx'\in\mathbb R^d : d_X(\bfx,\bfx') \leq t}.
\]
\end{definition}

\begin{definition}[Covering numbers and metric entropy]
Let $K, T\subseteq\mathbb R^d$ be two convex bodies. Then, the \emph{covering number} $E(K, T)$ is defined as
\[
    E(K, T) \coloneqq \min \braces*{k\in\mathbb N : \exists \{\bfx_i\}_{i=1}^k, K\subseteq \bigcup_{i=1}^k (\bfx_i + T)}.
\]
If $d_X$ is a metric and $t>0$ a radius, then $E(K, d_X, t)$ is defined as
\[
    E(K, d_X, t) \coloneqq E(K, B_X(0, t))
\]
(see \cref{def:dx-ball}). The \emph{metric entropy} is the logarithm of the covering number.
\end{definition}

Next, we introduce some notation that is specific to our setting of $\ell_p$ subspace embeddings.

\begin{definition}
For a matrix $\bfA\in\mathbb R^{n\times d}$ and $p\geq 1$, we define the ball
\[
    B^p(\bfA) \coloneqq \braces*{\bfA\bfx\in\mathbb R^n : \norm*{\bfA \bfx}_p \leq 1}.
\]
We simply write $B^p$ if $\bfA$ is clear from context.
\end{definition}

\subsection{Dual Sudakov Minoration}

One powerful tool for bounding covering numbers for covers of the Euclidean ball is the \emph{dual Sudakov minoration} theorem, which bounds covering numbers in terms of the so-called \emph{Levy mean}:

\begin{definition}[Levy mean]
Let $\norm*{\cdot}_X$ be a norm. Then, the \emph{Levy mean of $\norm*{\cdot}_X$} is defined to be
\[
    M_X \coloneqq \frac{\E_{\bfg\in\mathcal N(0,\bfI_d)}\norm*{\bfg}_X}{\E_{\bfg\in\mathcal N(0,\bfI_d)}\norm*{\bfg}_2}.
\]
\end{definition}

Bounds on the Levy mean imply bounds for covering the Euclidean ball by $\norm*{\cdot}_X$-balls via the following result:

\begin{theorem}[Dual Sudakov minoration, Proposition 4.2 of \cite{BLM1989}]
\label{thm:dual-sudakov}
Let $\norm*{\cdot}_X$ be a norm, and let $B\subseteq\mathbb R^d$ denote the Euclidean ball in $d$ dimensions. Then,
\[
    \log E(B, \norm*{\cdot}_X, t) \leq O(d)\frac{M_X^2}{t^2}
\]
\end{theorem}

\subsection{Entropy Estimates for \texorpdfstring{$p>2$}{p>2}}

We now use the preceding results to obtain the entropy estimates necessary to prove our main result for $p > 2$. We start by bounding the Levy mean for the norm defined by $\bfx\mapsto\norm*{\bfA\bfx}_q$ for some matrix $\bfA$. 

\begin{lemma}
Let $q\geq 2$ and let $\bfA\in\mathbb R^{n\times d}$. Let $\tau \geq \max_{i=1}^n\norm*{\bfe_i^\top\bfA}_2^2$. Then,
\[
    \E_{\bfg\sim\mathcal N(0,\bfI_d)}\norm{\bfA\bfg}_q \leq n^{1/q}\sqrt{q\cdot \tau}.
\]
\end{lemma}
\begin{proof}
We have for every $i\in[n]$ that
\[
    \E\abs{[\bfA\bfg](i)}^q = \frac{2^{q/2}\Gamma(\frac{q+1}{2})}{\sqrt{\pi}}\norm{\bfe_i^\top\bfA}_2^q \leq q^{q/2} \cdot \tau^{q/2}
\]
since $[\bfA\bfg](i)$ is distributed as a Gaussian random variable. Then by Jensen's inequality and linearity of expectation,
\[
    \E_{\bfg\sim\mathcal N(0,\bfI_d)}\norm{\bfA\bfg}_q \leq \parens*{\E_{\bfg\sim\mathcal N(0,\bfI_d)}\norm{\bfA\bfg}_q^q}^{1/q} =  \parens*{n\cdot q^{q/2} \cdot \tau^{q/2}}^{1/q} = n^{1/q}\sqrt{q\cdot \tau}
\]
\end{proof}

By combining the above calculation with \cref{thm:dual-sudakov}, we obtain the following:

\begin{corollary}\label{cor:2-q-cover}
Let $2\leq q < \infty$ and let $\bfA\in\mathbb R^{n\times d}$ be orthonormal. Let $\tau \geq \max_{i=1}^n\norm*{\bfe_i^\top\bfA}_2^2$. Then,
\[
    \log E(B^2, B^q, t) \leq O(1) \frac{n^{2/q} q\cdot \tau}{t^2}
\]
\end{corollary}
\begin{proof}
For $\bfA$ orthonormal, $B^2(\bfA) = B^2$ is isometric to the Euclidean ball in $d$ dimensions, and thus \cref{thm:dual-sudakov} applies.
\end{proof}

We also get a similar result for $q = \infty$, by applying \cref{cor:2-q-cover} with $q = O(\log n)$. 

\begin{corollary}\label{cor:2-inf-cover}
Let $\bfA\in\mathbb R^{n\times d}$ be orthonormal. Let $\tau \geq \max_{i=1}^n \norm*{\bfe_i^\top\bfA}_2^2$. Then,
\[
    \log E(B^2, B^\infty, t) \leq O(1) \frac{(\log n)\cdot \tau}{t^2}
\]
\end{corollary}
\begin{proof}
This follows from the fact that for $\bfy\in\mathbb R^n$,
\[
    \norm*{\bfy}_\infty \leq \norm*{\bfy}_q \leq n^{1/q} \norm*{\bfy}_\infty = O(1)\norm*{\bfy}_\infty
\]
for $q = O(\log n)$.
\end{proof}

\subsection{Entropy Estimates for \texorpdfstring{$p<2$}{p<2}}

By interpolation, we can improve the bound in \cref{cor:2-q-cover}, which is needed for our results for $p<2$:

\begin{lemma}\label{lem:2-r-cover}
Let $2 < r < \infty$ and let $\bfA\in\mathbb R^{n\times d}$ be orthonormal. Let $\tau \geq \max_{i=1}^n \norm*{\bfe_i^\top\bfA}_2^2$. Let $1 \leq t \leq \poly(d)$. Then,
\[
    \log E(B^2, B^r, t) \leq O(1) \frac{1}{(t/2)^{2r / (r-2)}}\cdot \parens*{\frac{r}{r-2}\log d + \log n}\tau
\]
\end{lemma}
\begin{proof}

Let $q > r$, and let $0<\theta<1$ satisfy
\[
    \frac1r = \frac{1-\theta}2 + \frac\theta{q}
\]
Then by H\"older's inequality, we have for any $\bfy\in\mathbb R^n$ that
\[
    \norm*{\bfy}_r = \parens*{\sum_{i=1}^n \abs{\bfy(i)}^{r(1-\theta)}\abs{\bfy(i)}^{r\theta}}^{1/r} \leq \parens*{\sum_{i=1}^n \abs{\bfy(i)}^2}^{(1-\theta)/2}\parens*{\sum_{i=1}^n \abs{\bfy(i)}^q}^{\theta/q} = \norm*{\bfy}_2^{1-\theta}\norm*{\bfy}_q^\theta
\]
Then for any $\bfy,\bfy'\in B^2$, we have
\[
    \norm*{\bfy-\bfy'}_r \leq \norm*{\bfy-\bfy'}_2^{1-\theta}\norm*{\bfy-\bfy'}_q^\theta \leq 2\norm*{\bfy-\bfy'}_q^\theta
\]
so
\[
    \log E(B^2, B^r, t) \leq \log E(B^2, B^q, (t/2)^{1/\theta}) \leq O(1)\frac{n^{2/q} q\cdot \tau}{(t/2)^{2/\theta}}
\]
by \cref{cor:2-q-cover}. Now, we have
\[
    \frac2\theta = 2\frac{\frac12 - \frac1q}{\frac12 - \frac1r} = \frac{q-2}{q}\frac{2r}{r-2}
\]
so by taking $q = O(\frac{r}{r-2}\log d + \log n)$, we have that $n^{2/q} = O(1)$ and $(t/2)^{1/\theta} = \Theta(1) (t/2)^{2r/(r-2)}$, so we conclude as claimed.
\end{proof}

Using \cref{lem:2-r-cover}, we obtain the following analogue of \cref{cor:2-q-cover} for $p<2$. 

\begin{lemma}\label{lem:p-inf-cover}
Let $1 \leq p < 2$ and let $\bfA\in\mathbb R^{n\times d}$ be orthonormal. Let $\tau \geq \max_{i=1}^n \norm*{\bfe_i^\top\bfA}_2^2$. Then,
\[
    \log E(B^p, B^\infty, t) \leq O(1) \frac1{t^p}\parens*{\frac{\log d}{2-p} + \log n}\tau.
\]
\end{lemma}
\begin{proof}
In order to bound a covering of $B^p$ by $B^\infty$, we first cover $B^p$ by $B^2$, and then use \cref{cor:2-inf-cover} to cover $B^2$ by $B^\infty$. 

We will first bound $E(B^p, B^2, t)$ using \cref{lem:2-r-cover}. For each $k\geq 0$, let $\mathcal E_k\subseteq B^p$ be a maximal subset of $B^p$ such that for each distinct $\bfy,\bfy'\in \mathcal E_k$, $\norm*{\bfy-\bfy'}_2 > 8^k t$, with $\mathcal E_k \coloneqq \{0\}$ for $8^{k+1}t > n^{1/p - 1/q}$. Note then that
\[
    \abs{\mathcal E_k} \geq E(B^p, B^2, 8^k t).
\]
By averaging, for each $k$, there exists $\bfy^{(k)}\in\mathcal E_k$ such that if
\[
    \mathcal F_k \coloneqq \braces*{\bfy\in\mathcal E_k : \norm{\bfy - \bfy^{(k)}}_2 \leq 8^{k+1}t},
\]
then
\[
    \abs{\mathcal F_k} \geq \frac{\abs{\mathcal E_k}}{E(B^p, B^2, 8^{k+1}t)} \geq \frac{E(B^p, B^2, 8^{k}t)}{E(B^p, B^2, 8^{k+1}t)}
\]
We now use this observation to construct an $\ell_{p'}$-packing of $B^2$, where $p'$ is the H\"older conjugate of $p$. Let 
\[
    \mathcal G_k \coloneqq \braces*{\frac{1}{8^{k+1}t}(\bfy - \bfy^{(k)}) : \bfy\in\mathcal F_k}.
\]
Then, $\mathcal G_k \subseteq B^2$ and $\mathcal G_k \subseteq B^p \cdot 2 / 8^{k+1}t$, and $\norm{\bfy - \bfy'}_2 > 1/8$ for every distinct $\bfy,\bfy'\in \mathcal G_k$. Then by H\"older's inequality,
\[
    \frac1{8^2} \leq \norm*{\bfy-\bfy'}_2^2 \leq \norm*{\bfy-\bfy'}_p \norm*{\bfy-\bfy'}_{p'} \leq \frac{4}{8^{k+1}t}\norm*{\bfy-\bfy'}_{p'}
\]
so $\norm*{\bfy-\bfy'}_{p'} \geq 2\cdot 8^{k-2}t$. Thus, $\mathcal G_k$ is an $\ell_{p'}$-packing of $B^2$, so
\begin{equation}\label{eq:dual-cover-bound}
    \log E(B^2, B^{p'}, 8^{k-2}t) \geq \log \abs{\mathcal G_k} = \log \abs{\mathcal F_k} \geq \log E(B^p, B^2, 8^{k}t) - \log E(B^p, B^2, 8^{k+1}t).
\end{equation}
Summing over $k$ gives
\begin{align*}
    \log E(B^p, B^2, t) &= \sum_{k\geq 0}\log E(B^p, B^2, 8^{k}t) - \log E(B^p, B^2, 8^{k+1}t) \\
    &\leq \sum_{k\geq 0} \log E(B^2, B^{p'}, 8^{k-2}t) && \text{\eqref{eq:dual-cover-bound}} \\
    &\leq O(1) \frac{1}{(t/2)^{2p' / (p'-2)}}\cdot \parens*{\frac{p'}{p'-2}\log d + \log n}\tau && \text{\cref{lem:2-r-cover} and \cref{cor:2-inf-cover}} \\
    &= O(1) \frac{1}{(t/2)^{2p / (2-p)}}\cdot \parens*{\frac{p}{2-p}\log d + \log n}\tau
\end{align*}
where we take $p' / (p'-2) = 1$ for $p' = \infty$. Using this and \cref{cor:2-inf-cover}, we now bound
\begin{align*}
    \log E(B^p, B^\infty, t) &\leq \log E(B^p, B^2, \lambda) + \log E(B^2, B^\infty, t/\lambda) \\
    &\leq O(1) \frac{1}{(\lambda/2)^{2p / (2-p)}}\cdot \parens*{\frac{p}{2-p}\log d + \log n}\tau + O(1) \frac{(\log n)\cdot \tau}{(t/\lambda)^2} \\
\end{align*}
for any $\lambda\in [1, t]$. We choose $\lambda$ satisfying
\[
    \frac1{(\lambda/2)^{2p/(2-p)}} = \frac{(\lambda/2)^2}{t^2},
\]
which gives
\[
    (\lambda/2)^{2p/(2-p)} = \parens*{t^2}^{\frac{2p/(2-p)}{2 + 2p/(2-p)}} = t^p
\]
so we obtain a bound of
\[
    O(1) \frac1{t^p}\parens*{\frac{1}{2-p}\log d + \log n}\tau.
\]
\end{proof}

\section{Bounding a Gaussian Process}

Using our estimates from \cref{sec:entropy-est}, we study estimates on the Gaussian process given by
\[
    X: \bfy \mapsto \sum_{i=1}^n  \bfg_i \abs{\bfy(i)}^p, \qquad \bfy \in B^p(\bfA)
\]
for $\bfg\sim\mathcal N(0,\bfI_n)$, and in particular, tail bounds and moment bounds on the quantity
\[
    \sup_{\bfy\in B^p(\bfA)} \abs{X(\bfy)} = \sup_{\norm*{\bfA\bfx}_p \leq 1}\abs*{\sum_{i=1}^n  \bfg_i \abs{[\bfA\bfx](i)}^p}
\]
As we show later, this Gaussian process bounds the error of the sensitivity sampling estimate. Our techniques here are based on similar results obtained by \cite{LT1991}.

The main tool is Dudley's tail inequality for Gaussian processes:
\begin{theorem}[Theorem 8.1.6, \cite{Ver2018}]\label{thm:dudley-tail}
Let $(X(t))_{t\in T}$ be a Gaussian process with pseudo-metric $d_X(s,t)\coloneqq \norm*{X(s) - X(t)}_2 = \sqrt{\E (X(s) - X(t))^2}$ and let
\[
    \diam(T) \coloneqq \sup\braces*{d_X(s,t) : s, t\in T}
\]
Then, there is a constant $C = O(1)$ such that for every $z\geq 0$,
\[
	\Pr\braces*{\sup_{t\in T}X_t \geq C\bracks*{\int_0^\infty \sqrt{\log E(T, d_X, u)}~du + z\cdot \diam(T)}} \leq 2\exp(-z^2)
\]
\end{theorem}

The expectation bound version of the theorem will also be useful:

\begin{theorem}[Theorem 8.1.3, \cite{Ver2018}]\label{thm:dudley}
Let $(X(t))_{t\in T}$ be a Gaussian process with pseudo-metric $d_X(s,t)\coloneqq \norm*{X(s) - X(t)}_2 = \sqrt{\E (X(s) - X(t))^2}$. Then, there is a constant $C = O(1)$ such that
\[
    \E\sup_{t\in T}X_t \leq C\int_0^\infty \sqrt{\log E(T, d_X, u)}~du.
\]
\end{theorem}

\subsection{Bounds on \texorpdfstring{$d_X$}{dX}}

We first bound $d_X$ as well as the $d_X$-diameter of $B^p(\bfA)$. 

\begin{lemma}
\label{lem:dx-bound}
Let $1\leq p < \infty$ and let $\bfA\in\mathbb R^{n\times d}$. Define the pseudo-metric
\[
    d_X(\bfy,\bfy') \coloneqq \parens*{\E_{\bfg\sim\mathcal N(0,\bfI_n)}\abs*{\sum_{i=1}^n  \bfg_i \abs{\bfy(i)}^p - \sum_{i=1}^n  \bfg_i \abs{\bfy'(i)}^p}^2}^{1/2}
\]
Let $\sigma \geq \max_{i=1}^n \bfsigma_i^p(\bfA)$. Then, for any $\bfy,\bfy'\in B_p(\bfA)$,
\[
    d_X(\bfy,\bfy') \leq \begin{dcases}
        2\norm*{\bfy-\bfy'}_\infty^{p/2} & p < 2 \\
        2p \cdot \sigma^{1/2-1/p}\cdot \norm*{\bfy-\bfy'}_\infty & p > 2
    \end{dcases}
\]
\end{lemma}
\begin{proof}
Note first that by expanding out the square and noting that $\E[\bfg_i\bfg_j] = \mathbbm{1}(i=j)$, we have
\[
    d_X(\bfy,\bfy') = \parens*{\sum_{i=1}^n  (\abs{\bfy(i)}^p - \abs{\bfy'(i)}^p)^2}^{1/2}
\]
For $p<2$, we bound this as
\begin{align*}
    d_X(\bfy,\bfy')^2 &= \sum_{i=1}^n  (\abs{\bfy(i)}^p - \abs{\bfy'(i)}^p)^2 \\
    &= \sum_{i=1}^n  (\abs{\bfy(i)}^{p/2} - \abs{\bfy'(i)}^{p/2})^2(\abs{\bfy(i)}^{p/2} + \abs{\bfy'(i)}^{p/2})^2 \\
    &\leq \sum_{i=1}^n  (\abs{\bfy(i) - \bfy'(i)}^{p/2})^2(\abs{\bfy(i)}^{p/2} + \abs{\bfy'(i)}^{p/2})^2 \\
    &\leq 2\norm*{\bfy-\bfy'}_\infty^p\sum_{i=1}^n \abs{\bfy(i)}^{p} + \abs{\bfy'(i)}^{p} \\
    &\leq 4\norm*{\bfy-\bfy'}_\infty^p.
\end{align*}
For $p>2$, we have by convexity that
\[
    \abs{\bfy(i)}^p - \abs{\bfy'(i)}^p \leq p\abs{\bfy(i) - \bfy'(i)}(\abs{\bfy(i)}^{p-1} + \abs{\bfy'(i)}^{p-1})
\]
and that $\norm*{\bfy}_\infty \leq \sigma^{1/p}$, so we have
\begin{align*}
    d_X(\bfy,\bfy')^2 &= \sum_{i=1}^n  (\abs{\bfy(i)}^p - \abs{\bfy'(i)}^p)^2 \\
    &\leq p^2 \sum_{i=1}^n  \abs{\bfy(i) - \bfy'(i)}^2(\abs{\bfy(i)}^{p-1} + \abs{\bfy'(i)}^{p-1})^2 \\
    &\leq 2p^2 \norm*{\bfy-\bfy'}_\infty^2\sum_{i=1}^n \abs{\bfy(i)}^{2p-2} + \abs{\bfy'(i)}^{2p-2} \\
    &\leq 2p^2 \max\{\norm*{\bfy}_\infty,\norm{\bfy'}_\infty\}^{p-2}\norm*{\bfy-\bfy'}_\infty^2\sum_{i=1}^n \abs{\bfy(i)}^{p} + \abs{\bfy'(i)}^{p} \\
    &\leq 4p^2 \sigma^{1-2/p}\norm*{\bfy-\bfy'}_\infty^2 \\
\end{align*}
\end{proof}

\begin{lemma}
\label{lem:diam}
Let $1 \leq p < \infty$ and let $\bfA\in\mathbb R^{n\times d}$. Let $\sigma \geq \max_{i=1}^n \bfsigma_i^p(\bfA)$. Then, the diameter of $B^p(\bfA)$ with respect to $d_X$ is bounded by
\[
    \diam(B^p(\bfA)) \leq \begin{cases}
        4\cdot\sigma^{1/2} & p < 2 \\
        4p \cdot \sigma^{1/2} & p > 2
    \end{cases}
\]
\end{lemma}
\begin{proof}
For any $\bfy\in B^p(\bfA)$, we have that $\norm*{\bfy}_\infty \leq \sigma^{1/p}$, so combining the triangle inequality and \cref{lem:dx-bound} yields the result.
\end{proof}

\subsection{Computing the Entropy Integral}
\label{sec:ent-int}

We may now evaluate the entropy integral required in \cref{thm:dudley-tail}. We use the following calculus lemma:

\begin{lemma}\label{lem:calc}
Let $0 < \lambda \leq 1$. Then,
\[
    \int_0^\lambda \sqrt{\log\frac1t}~dt = \lambda \sqrt{\log(1/\lambda)} + \frac{\sqrt\pi}{4}\erfc(\sqrt{\log(1/\lambda)}) \leq \lambda \parens*{\sqrt{\log(1/\lambda)} + \frac{\sqrt\pi}2}
\]
\end{lemma}
\begin{proof}
We calculate
\begin{align*}
    \int_0^\lambda \sqrt{\log\frac{1}{t}}~dt &= 2\int_{\sqrt{\log(1/\lambda)}}^\infty x^2 \exp(-x^2)~dx && \text{$x = \sqrt{\log(1/t)}$} \\
    &= -\int_{\sqrt{\log(1/\lambda)}}^\infty x \cdot -2x \exp(-x^2)~dx \\
    &= -\parens*{x\exp(-x^2)\Big\vert^{\infty}_{\sqrt{\log(1/\lambda)}} - \int_{\sqrt{\log(1/\lambda)}}^\infty \exp(-x^2)~dx} && \text{integration by parts} \\
    &= \lambda  \sqrt{\log\frac1\lambda} + \frac{\sqrt\pi}{2}\erfc\parens*{\sqrt{\log\frac1\lambda}}
\end{align*}
\end{proof}

\begin{lemma}[Entropy integral bound for $p<2$]\label{lem:entropy-int-p<2}
Let $1 \leq p < 2$ and let $\bfA\in\mathbb R^{n\times d}$ be orthonormal. Let $\tau \geq \max_{i=1}^n \norm*{\bfe_i^\top\bfA}_2^2$ and let $\sigma \geq \max_{i=1}^n \bfsigma_i^p(\bfA)$. Then,
\[
    \int_0^\infty \sqrt{\log E(B^p, d_{X}, t)}~dt \leq O(\tau^{1/2})\parens*{\frac{\log d}{2-p} + \log n}^{1/2}\log\frac{d\sigma}{\tau}
\]
\end{lemma}
\begin{proof}
Note that it suffices to integrate the entropy integral to $\diam(B^p(\bfA))$ rather than $\infty$, which is at most $4\sigma^{1/2}$ for $p<2$ and $4p\sigma^{1/2}$ for $p>2$ by \cref{lem:diam}.

By \cref{lem:dx-bound}, we have that
\[
\log E(B^p, d_{X}, t) \leq \log E(B^p, 2\norm*{\cdot}_\infty^{p/2}, t) = \log E(B^p, B^\infty, (t/2)^{2/p})
\]

For small radii less than $\lambda$ for a parameter $\lambda$ to be chosen, we use a standard volume argument, which shows that
\[
    \log E(B^p, B^\infty, t) \leq O(d)\log\frac{n}{t}
\]
so
\begin{align*}
    \int_0^\lambda \sqrt{\log E(B^p, B^\infty, t)}~dt = \int_0^\lambda \sqrt{d\log\frac{n}{t}}~dt &\leq \lambda \sqrt{d\log n} + \sqrt d\int_0^\lambda \sqrt{\log\frac{1}{t}}~dt \\
    &\leq \lambda \sqrt{d\log n} + \sqrt d\parens*{\lambda \sqrt{\log\frac1\lambda} + \frac{\sqrt\pi}{2}\lambda} && \text{\cref{lem:calc}} \\
    &\leq O(\lambda)\sqrt{d\log\frac{n}{\lambda}}
\end{align*}
On the other hand, for large radii larger than $\lambda$, we use the bounds of \cref{lem:p-inf-cover}, which gives
\[
    \log E(B^p, B^\infty, (t/2)^{2/p}) \leq O(1) \frac1{t^2}\parens*{\frac{\log d}{2-p} + \log n}\tau
\]
so the entropy integral gives a bound of
\[
    O(1)\bracks*{\parens*{\frac{\log d}{2-p} + \log n}\tau}^{1/2}\int_\lambda^{4p\sigma^{1/2}} \frac1t~dt = O(1)\bracks*{\parens*{\frac{\log d}{2-p} + \log n}\tau}^{1/2}\log\frac{4p\sigma^{1/2}}{\lambda}.
\]
We choose $\lambda = \sqrt{\tau / d}$, which yields the claimed conclusion.
\end{proof}

An analogous result and proof holds for $p>2$.

\begin{lemma}[Entropy integral bound for $p>2$]
\label{lem:entropy-int-p>2}
Let $2 < p < \infty$ and let $\bfA\in\mathbb R^{n\times d}$ be orthonormal. Let $\tau \geq \max_{i=1}^n \norm*{\bfe_i^\top\bfA}_2^2$ and let $\sigma \geq \max_{i=1}^n \bfsigma_i^p(\bfA)$. Then,
\[
    \int_0^\infty \sqrt{\log E(B^p, d_{X}, t)}~dt \leq O(p\tau^{1/2})\cdot (\sigma n)^{1/2-1/p} (\log n)^{1/2}\cdot \log\frac{p^2 d\sigma}{\tau}
\]
\end{lemma}
\begin{proof}
The proof is similar to the case of $p<2$. We again introduce a parameter $\lambda$. For radii below $\lambda$, the bound is the same as \cref{lem:entropy-int-p<2}. For radii above $\lambda$, we use \cref{lem:dx-bound} to bound
\[
\log E(B^p, d_X, t) \leq \log E(B^p, 2p \cdot \sigma^{1/2-1/p}\cdot \norm*{\cdot}_\infty, t) \leq \log E(B^p, B^\infty, t / 2p \cdot \sigma^{1/2-1/p})
\]
Then by \cref{cor:2-inf-cover},
\begin{align*}
    \log E(B^p, B^\infty, t / 2p \cdot \sigma^{1/2-1/p}) &\leq \log E(B^2, B^\infty, t / 2p \cdot (\sigma n)^{1/2-1/p}) \\
    &\leq O(p^2) \frac{(\log n)\cdot \tau}{t^2} \cdot (\sigma n)^{1-2/p}
\end{align*}
so the entropy integral gives a bound of
\[
    O(p\tau^{1/2})\cdot (\sigma n)^{1/2-1/p} (\log n)^{1/2}\cdot \int_\lambda^{\diam(B^p(\bfA))}\frac{1}{t} ~dt \leq O(p\tau^{1/2})\cdot (\sigma n)^{1/2-1/p} (\log n)^{1/2}\cdot \log\frac{p\sigma^{1/2}}{\lambda}
\]
Choosing $\lambda = \sqrt{\tau/d}$ yields the claimed conclusion.
\end{proof}

\subsection{Moment Bounds}

Finally, using our tail bound from combining \cref{thm:dudley-tail} with the entropy bounds of \cref{sec:ent-int} and \cref{lem:diam}, we obtain the following moment bounds:

\begin{lemma}
\label{lem:moment-bound}
Let $\bfA\in\mathbb R^{n\times d}$ and $1\leq p < \infty$. Let $\tau \geq \max_{i=1}^n \bftau_i(\bfA)$ and let $\sigma \geq \max_{i=1}^n \bfsigma_i^p(\bfA)$. Let
\[
    \Lambda \coloneqq \sup_{\norm*{\bfA\bfx}_p \leq 1}\abs*{\sum_{i=1}^n  \bfg_i \abs{[\bfA\bfx](i)}^p}
\]
Let $\mathcal E \coloneqq \int_0^\infty \sqrt{\log E(B^p(\bfA), d_X, u)}~du$ and $\mathcal D = \diam(B^p(\bfA))$. Then,
\[
    \E_{\bfg\sim\mathcal N(0,\bfI_n)}[\abs{\Lambda}^l] \leq (2\mathcal E)^l (\mathcal E/\mathcal D) + O(\sqrt l \mathcal D)^{l}
\]
\end{lemma}
\begin{proof}
By \cref{thm:dudley-tail}, we have for $T = B^p(\bfA)$ that
\[
    \Pr\braces*{\Lambda \geq C\bracks*{\int_0^\infty \sqrt{\log E(T, d_X, u)}~du + z\cdot\diam(T)}} \leq 2\exp(-z^2)
\]
for a constant $C = O(1)$. Then,
\begin{align*}
    \E[(\Lambda / \mathcal D)^l] &= l\int_0^\infty z^l \Pr\{\Lambda \geq z \mathcal D\}~dz \\
    &\leq (2\mathcal E/\mathcal D)^{l+1} + l\int_{2\mathcal E/\mathcal D}^\infty z^l \Pr\{\Lambda \geq z\mathcal D\}~dz \\
    &\leq (2\mathcal E/\mathcal D)^{l+1} + l\int_{2\mathcal E/\mathcal D}^\infty z^l \Pr\{\Lambda \geq \mathcal E + (z/2)\mathcal D\}~dz \\
    &\leq (2\mathcal E/\mathcal D)^{l+1} + 2l\int_0^\infty z^l \exp(-z^2 / 4)~dz \\
    &\leq (2\mathcal E/\mathcal D)^{l+1} + O(l)^{l/2}
\end{align*}
so
\[
    \E[\Lambda^l] \leq (2\mathcal E)^l (\mathcal E/\mathcal D) + O(\sqrt l \mathcal D)^{l}.
\]
\end{proof}

\section{Sampling Guarantees}

We first reduce our proof of sampling guarantees to the problem of bounding a Gaussian process:
\begin{lemma}[Reduction to Gaussian processes]
\label{lem:gp-reduction}
Let $\bfA\in\mathbb R^{n\times d}$ and $1\leq p < \infty$. Let $\bfS$ be a random $\ell_p$ sampling matrix (\cref{def:sampling-matrix}). Then,
\[
    \E_{\bfS}\sup_{\norm*{\bfA\bfx}_p = 1} \abs*{\norm*{\bfS\bfA\bfx}_p^p - 1}^l \leq (2\pi)^{l/2} \E_{\bfS}\E_{\bfg\sim\mathcal N(0,\bfI_n)}\sup_{\norm*{\bfA\bfx}_p = 1} \abs*{\sum_{i\in S} \bfg_i \abs*{[\bfS\bfA\bfx](i)}^p}^l,
\]
where $S\subseteq[n]$ is the set of rows with sampling probability $q_i < 1$.
\end{lemma}
\begin{proof}
By a standard symmetrization argument \cite{CP2015, CD2021}, we have that
\[
    \E_{\bfS}\sup_{\norm*{\bfA\bfx}_p = 1} \abs*{\norm*{\bfS\bfA\bfx}_p^p - 1}^l \leq 2^l \E_{\bfS,\bfeps}\sup_{\norm*{\bfA\bfx}_p \leq 1} \abs*{\sum_{i\in S} \bfeps_i \abs*{[\bfS\bfA\bfx](i)}^p}^l,
\]
where $\bfeps\sim\{\pm1\}^n$ are independent Rademacher variables. In turn, the right hand side is bounded by
\[
    2^l (\pi/2)^{l/2}\E_{\bfS,\bfg}\sup_{\norm*{\bfA\bfx}_p \leq 1} \abs*{\sum_{i\in S} \bfg_i \abs*{[\bfS\bfA\bfx](i)}^p}^l
\]
via a Rademacher-Gaussian comparison theorem (see, e.g., Equation 4.8 of \cite{LT1991}).
\end{proof}

\subsection{Sensitivity Sampling, \texorpdfstring{$p<2$}{p<2}}

Our first result is a sensitivity sampling guarantee for $p<2$.

\begin{theorem}[Sensitivity Sampling for $p<2$]
\label{thm:sens-sample-p<2}
Let $\bfA\in\mathbb R^{n\times d}$ and $1\leq p < 2$. Let $\bfS$ be a random $\ell_p$ sampling matrix with sampling probabilities $q_i = \min\{1, 1/n + \bfsigma_i^p(\bfA) / \alpha\}$ for an oversampling parameter $\alpha$ set to
\begin{align*}
    \frac1\alpha &= \frac{\frS^p(\bfA)^{2/p-1}}{\eps^2} \bracks*{O(l\log n)^{2/p-1}\parens*{\frac{\log d}{2-p} + \log \frac{l\log n}{\eps}}(\log d)^2 + l} \\
    &= \frac{\frS^p(\bfA)^{2/p-1}}{\eps^2} \poly\parens*{\log n, \log\frac1\delta, \frac1{2-p}}
\end{align*}
for
\[
    l = O\parens*{\log\frac1\delta + \log\log n + \log \frac1{2-p} + \log\frac{d}{\eps}}.
\]
Then, with probability at least $1-\delta$, simultaneously for all $\bfx\in\mathbb R^d$,
\[
    \norm*{\bfS\bfA\bfx}_p^p = (1\pm\eps)\norm*{\bfA\bfx}_p^p.
\]
Furthermore, with probability at least $1-\delta$, $\bfS$ samples
\[
    \frac{\frS^p(\bfA)^{2/p}}{\eps^2} \poly\parens*{\log n, \log\frac1\delta, \log\frac1{2-p}}
\]
rows.
\end{theorem}
\begin{proof}
Our approach is to bound
\[
    \E_{\bfS}\sup_{\norm*{\bfA\bfx}_p = 1} \abs*{\norm*{\bfS\bfA\bfx}_p^p - 1}^l
\]
for a large even integer $l$. Using \cref{lem:gp-reduction}, we first bound
\[
    \E_{\bfS}\sup_{\norm*{\bfA\bfx}_p = 1} \abs*{\norm*{\bfS\bfA\bfx}_p^p - 1}^l \leq (2\pi)^{l/2} \E_{\bfS}\E_{\bfg\sim\mathcal N(0,\bfI_n)}\sup_{\norm*{\bfA\bfx}_p = 1} \abs*{\sum_{i\in S} \bfg_i \abs*{[\bfS\bfA\bfx](i)}^p}^l
\]
where $S = \{i\in[n] : q_i < 1\}$. For simplicity of presentation, we assume $S = [n]$, which will not affect our proof.

By \cref{thm:lewis-weight-sampling}, there exists a matrix $\bfA'\in\mathbb R^{m_1\times d}$ with $m_1 = O(d(\log d)^3)$ such that
\[
    \norm{\bfA'\bfx}_p^p = (1\pm1/2)\norm{\bfA\bfx}_p^p
\]
for all $\bfx\in\mathbb R^d$. Furthermore, because $\bfA'$ in \cref{thm:lewis-weight-sampling} is constructed by random sampling, \cref{lem:sampling-preserve-sens} shows that $\frS^p(\bfA') \leq 8\frS^p(\bfA)$ (note that we only need existence of this matrix). We then construct a matrix $\bfA''\in\mathbb R^{m_2 \times d}$ with $m_2 = O(\alpha^{-1}\frS^p(\bfA) + d(\log d)^3) = O(\alpha^{-1}\frS^p(\bfA))$ such that
\[
    \sigma \coloneqq \max_{i=1}^n \bfsigma_i^p(\bfA'') \leq \alpha,
\]
$\frS^p(\bfA') = \frS^p(\bfA'')$, and $\norm{\bfA'\bfx}_p = \norm{\bfA''\bfx}_p$ for all $\bfx\in\mathbb R^d$ by viewing $\bfA'$ as an $(m_1 + \alpha^{-1}\frS^p(\bfA)) \times d$ matrix with all zeros except for the first $m_1$ rows and then applying \cref{lem:sens-flat}. 

Now let
\[
    \bfA''' \coloneqq \begin{pmatrix}
    \bfA'' \\ \bfS\bfA
    \end{pmatrix}
\]
be the $(m_2 + n_\bfS) \times d$ matrix formed by the vertical concatenation of $\bfA''$ with $\bfS\bfA$, where $n_\bfS$ is the number of rows sampled by $\bfS$. 

\paragraph{Sensitivity Bounds for $\bfA'''$.}

We will first bound the $\ell_p$ sensitivities of $\bfA'''$. For any row $i$ corresponding to a row of $\bfA''$, the $\ell_p$ sensitivities are already bounded by $\alpha$, and furthermore, $\ell_p$ sensitivities can clearly only decrease with row additions. For any row $i$ corresponding to a row of $\bfS\bfA$ that is sampled with probability $q_i < 1$, we have that
\[
    \frac{\abs*{[\bfS\bfA\bfx](i)}^p}{\norm*{\bfA'''\bfx}_p^p} \leq 2\frac{\abs*{[\bfS\bfA\bfx](i)}^p}{\norm*{\bfA\bfx}_p^p} \leq 2\frac1{q_i}\frac{\abs*{[\bfA\bfx](i)}^p}{\norm*{\bfA\bfx}_p^p} \leq 2\frac{\bfsigma_i^p(\bfA)}{q_i} = 2\alpha.
\]
Thus, we have that $\bfsigma_i^p(\bfA''') \leq 2\alpha$ for every row $i$ of $\bfA'''$.

With a bound on the $\ell_p$ sensitivities of $\bfA'''$ in hand, we may then convert this into a bound on the leverage scores of $\bfA'''$ using \cref{lem:sens-mon-rev}, which gives
\[
    \tau \coloneqq \max_{i=1}^n \bftau_i(\bfA''') \leq (2\alpha)^{2/p}(m_2 + n_\bfS)^{2/p-1}
\]
where $n_\bfS$ is the number of nonzero entries of $\bfS$.

\paragraph{Moment Bounds on the Sampling Error.}

We now fix a choice of $\bfS$, and define
\[
F_{\bfS} \coloneqq \sup_{\norm*{\bfA\bfx}_p = 1} \abs*{\norm*{\bfS\bfA\bfx}_p^p - 1}.
\]
Note that the event that $n_\bfS$ is at least
\[
    n_{\mathrm{thresh}} \coloneqq O(l\log n)\E[n_S] = O(l\log n) \alpha^{-1}\frS^p(\bfA),
\]
occurs with probability at most $\poly(n)^{-l}$ by Chernoff bounds over the randomness of $\bfS$, and
\[
    F_{\bfS}^l \leq \bracks*{1 + \sum_{i=1}^n \frac1{q_i}}^l \leq (n+1)^{2l},
\]
and thus this event contributes at most $\poly(n)^{-l}$ to the moment bound $\E F_{\bfS}^l$. Thus, we focus on bounding $\E F_\bfS^l$ conditioned on $n_\bfS \leq n_{\mathrm{thresh}}$. Define
\[
    G_\bfS \coloneqq \sup_{\norm*{\bfA'''\bfx}_p = 1} \abs*{\sum_{i=1}^{m_2 + n_\bfS} \bfg_i \abs*{[\bfA'''\bfx](i)}^p}
\]
for $\bfg\sim\mathcal N(0,\bfI_{m_2 + n_\bfS})$. Then,
\[
    \norm*{\bfA'''\bfx}_p^p \leq (1 + 2 + F_{\bfS})\norm*{\bfA\bfx}_p^p
\]
so
\begin{equation}
\label{eq:Fsl-bound}
\begin{aligned}
    F_\bfS^l &\leq 2^l\sup_{\norm*{\bfA\bfx}_p = 1} \abs*{\sum_{i=1}^{m_2 + n_\bfS} \bfg_i \abs*{[\bfA'''\bfx](i)}^p}^l \\
    &\leq 2^l(1 + 2 + F_{\bfS})^l\sup_{\norm*{\bfA'''\bfx}_p = 1} \abs*{\sum_{i=1}^{m_2 + n_\bfS} \bfg_i \abs*{[\bfA'''\bfx](i)}^p}^l \\
    &\leq 2^{2l-1}(3^l+F_\bfS^l) G_\bfS^l.
\end{aligned}
\end{equation}
We then take expectations on both sides with respect to $\bfg\sim\mathcal N(0,\bfI_{m_2+n_\bfS})$, and bound the right hand side using \cref{lem:moment-bound}, which gives
\[
    \E_{\bfg\sim\mathcal N(0,\bfI_{m_2+n_\bfS})}G_\bfS^l \leq \parens*{2\mathcal E}^l\frac{\mathcal E}{\mathcal D} + O(\sqrt l \mathcal D)^l
\]
where $\mathcal E$ is the entropy integral and $\mathcal D = 4\sigma^{1/2}$ is the diameter by \cref{lem:diam}. We have by \cref{lem:entropy-int-p<2} that
\begin{align*}
    \mathcal E &\leq O(\tau^{1/2})\parens*{\frac{\log d}{2-p} + \log(m_2+n_\bfS)}^{1/2}\log\frac{d\sigma}{\tau} \\
    &\leq O(\alpha^{1/p}(m_2 + n_\bfS)^{1/p-1/2})\parens*{\frac{\log d}{2-p} + \log(m_2+n_\bfS)}^{1/2}\log\frac{d\sigma}{\tau}
\end{align*}
Thus, conditioned on $n_\bfS \leq n_{\mathrm{thresh}}$, we have that
\[
    \E_{\bfg\sim\mathcal N(0,\bfI_{m_2+n_\bfS})}G_\bfS^l \leq \bracks*{O(\alpha^{1/p}n_{\mathrm{thresh}}^{1/p-1/2})\parens*{\frac{\log d}{2-p} + \log n_{\mathrm{thresh}}}^{1/2}\log d}^l + O(\sqrt l \sqrt\alpha)^l.
\]
Note that
\[
    \alpha^{1/p}n_{\mathrm{thresh}}^{1/p-1/2} = O(l\log n)^{1/p-1/2} \alpha^{1/p} (\alpha^{-1}\frS^p(\bfA))^{1/p-1/2} = O(l\log n)^{1/p-1/2} \alpha^{1/2} \frS^p(\bfA)^{1/p-1/2},
\]
which shows that
\[
    \E_{\bfg\sim\mathcal N(0,\bfI_{m_2+n_\bfS})} G_\bfS^l \leq \eps^l \delta
\]
due to our choice of $\alpha$ and $l$.

Now if we take conditional expectations on both sides of \eqref{eq:Fsl-bound} conditioned on the event $\mathcal F$ that $n_\bfS \leq n_{\mathrm{thresh}}$, then we have
\[
    \E[F_\bfS^l \mid \mathcal F] \leq 2^{2l-1}(3^l + \E[F_\bfS^l \mid \mathcal F])\eps^l \delta \leq (3^l + \E[F_\bfS^l \mid \mathcal F])(4\eps)^l \delta
\]
which means
\[
    \E[F_\bfS^l \mid \mathcal F] \leq \frac{(12\eps)^l \delta}{1 - (4\eps)^l\delta} \leq 2(12\eps)^l \delta
\]
for $(4\eps)^l\delta \leq 1/2$. We thus have
\[
    \E[F_\bfS^l] \leq \frac{(12\eps)^l \delta}{1 - (4\eps)^l\delta} \leq 2(12\eps)^l \delta + \poly(n)^{-l}
\]
altogether. Finally, we have by a Markov bound that
\[
    F_\bfS^l \leq 2(12\eps)^l + \frac1\delta \poly(n)^l \leq 3(12\eps)^l
\]
with probability at least $1-\delta$, which means that
\[
    F_\bfS \leq 3\cdot 12\eps = 36\eps
\]
with probability at least $1-\delta$. Rescaling $\eps$ by constant factors yields the claimed result.
\end{proof}

\subsection{Sensitivity Sampling, \texorpdfstring{$p>2$}{p>2}}

For $p>2$, we first need a construction of a matrix with a small number of rows and small sensitivity. While this construction can be made to be a randomized algorithm succeeding with high probability, it uses a sophisticated recursive sampling strategy. While this is necessary for our results later in Theorem \ref{thm:recursive-sens-lev-sampling}, such a complicated algorithm may be undesirable. In \cref{thm:sens-sample-p>2}, we use this result to show that a more direct one-shot sensitivity sampling can in fact achieve a similar guarantee.

\begin{lemma}[Recursive Sensitivity Sampling]
\label{lem:recursive-sens-sampling}
Let $\bfA\in\mathbb R^{n\times d}$ and $2 < p < \infty$. Let $0<\eps<1$. Then, there exists a matrix $\bfA'\in\mathbb R^{m\times d}$ for 
\[
    m = O(p^2) \frac{\frS^p(\bfA)^{2-2/p}}{\eps^2}\log(pd)^2 \log \frac{pd}{\eps}
\]
such that
\[
    \norm*{\bfA'\bfx}_p^p = (1\pm\eps)\norm*{\bfA\bfx}_p^p
\]
for every $\bfx\in\mathbb R^d$ and $\frS^p(\bfA') \leq (1+O(\eps))\frS^p(\bfA)$.
\end{lemma}
\begin{proof}
Let $\bfA'\in\mathbb R^{m\times d}$ be the flattened isometric matrix given by \cref{lem:sens-flat} with $C = 4$, where $m \leq (5/4)n$. Then for all $i\in[m]$, we have that
\[
    \bfsigma_i^p(\bfA') \leq 4\frac{\frS^p(\bfA)}{n} \leq 5\frac{\frS^p(\bfA')}{m}.
\]

Now consider the random sampling matrix $\bfS$ with sampling probabilities $q_i = 1/2$. Note then that sampling with probability $q_i = 1/2$ and scaling by $1/q_i = 2$ corresponds to muliplying by the random variable $\bfeps_i + 1$, where $\bfeps_i$ is a Rademacher variable. Thus,
\[
    \E_\bfS\sup_{\norm*{\bfA'\bfx}_p = 1} \abs*{\norm*{\bfS\bfA'\bfx}_p^p - 1} = \E_{\bfeps}\sup_{\norm*{\bfA'\bfx}_p = 1} \abs*{\sum_{i=1}^n (\bfeps_i+1)\abs*{[\bfA'\bfx](i)}^p - \abs*{[\bfA'\bfx](i)}^p} = \E_{\bfeps}\sup_{\norm*{\bfA'\bfx}_p = 1} \abs*{\sum_{i=1}^n \bfeps_i\abs*{[\bfA'\bfx](i)}^p}.
\]
By \cref{lem:gp-reduction} and \cref{thm:dudley}, this is bounded by
\[
    O(1)\int_0^\infty \sqrt{\log E(T, d_X, u)}~du \leq O(p\tau^{1/2})\cdot(\sigma n)^{1/2-1/p}(\log n)^{1/2}\cdot \log \frac{p^2 d \sigma}{\tau}
\]
where $\tau$ is an upper bound on the leverage scores of $\bfA'$ and $\sigma$ is an upper bound on the $\ell_p$ sensitivities of $\bfA'$. By \cref{lem:sens-mon}, we have that $\tau\leq \sigma$, and furthermore, we can take $\sigma = 5\frS^p(\bfA')/m$. Thus, the resulting bound on the expected sampling error is at most
\[
    \eps_\bfA \coloneqq O(p) \frac{\frS^p(\bfA)^{1-1/p}}{\sqrt n}(\log n)^{1/2}\log(pd)
\]
so with probability at least $99/100$, the same bound holds up to a factor of $100$. Furthermore, $\bfS$ samples $m/2 \leq (5/8)n$ rows in expectation, so by Markov's inequality, it samples at most $(3/2) m/2 \leq (15/16) n$ rows with probability at least $1/3$. We also have that
\[
    \frac{\bfsigma_i^p(\bfA)}{q_i} = 2\bfsigma_i^p(\bfA') \leq 10\frac{\frS^p(\bfA')}{m} 
\]
so by \cref{lem:sampling-preserve-sens-hp}, we have that
\[
    \Pr\braces*{\frS^p(\bfS\bfA') = (1\pm O(\eps_\bfA))\frS^p(\bfA)} \geq \frac{99}{100}.
\]
By a union bound, $\bfS\bfA'$ samples at most $(15/16)n$ rows, has sampling error at most $\eps_n$, and has $\ell_p$ total sensitivity at most $(1+O(\eps_\bfA))\frS^p(\bfA)$ with probability at least $1/3 - 1/100 - 1/100 > 0$. Thus, such an instantiation of $\bfS\bfA'$ exists. 

We now recursively apply our reasoning, by repeatedly applying the flattening and sampling operation. Note that each time we repeat this procedure, the number of rows goes down by a factor of $15/16$, while the total sensitivity and total sampling error accumulates. Let $\bfA_l$ denote the matrix obtained after $l$ recursive applications of this procedure and let $n_l$ denote the number of rows of $\bfA_l$. Then,
\begin{align*}
    \eps_{\bfA_{l+1}} &= O(p) \frac{\frS^p(\bfA_{l+1})^{1-1/p}}{\sqrt n_{l+1}}(\log n_{l+1})^{1/2}\log(pd) \\
    &\geq (1-O(\eps_{\bfA_l})) O(p) \frac{\frS^p(\bfA_{l})^{1-1/p}}{\sqrt n_{l+1}}(\log n_{l+1})^{1/2}\log(pd) \\
    &\geq \sqrt{\frac{16}{15}} (1-O(\eps_{\bfA_l})) O(p) \frac{\frS^p(\bfA_{l})^{1-1/p}}{\sqrt n_{l}}(\log n_{l})^{1/2}\log(pd) \\
    &\geq \frac{101}{100} \cdot \eps_{\bfA_l}
\end{align*}
as long as $\eps_{\bfA_l}$ is less than some absolute constant. Thus, the sum of the $\eps_{\bfA_{l}}$ are dominated by the last $\eps_{\bfA_l}$, up to a constant factor. Now let $L$ be the smallest integer $l$ such that $\eps_{\bfA_l} \leq \eps$. Then, we have that
\[
    \frS^p(\bfA_L) \leq (1+O(\eps))\frS^p(\bfA)
\]
and thus
\[
    \norm*{\bfA_L\bfx}_p^p = (1\pm O(\eps))\norm*{\bfA\bfx}_p^p
\]
for every $\bfx\in\mathbb R^d$. Furthermore, $n_L$ satisfies
\[
    \eps = O(p) \frac{\frS^p(\bfA)^{1-1/p}}{\sqrt n_L}(\log n_L)^{1/2}\log(pd)
\]
or
\[
    n_L = O(p^2) \frac{\frS^p(\bfA)^{2-2/p}}{\eps^2}\log(pd)^2 \log \frac{pd}{\eps}.
\]
\end{proof}

\begin{theorem}[Sensitivity Sampling for $p>2$]
\label{thm:sens-sample-p>2}
Let $\bfA\in\mathbb R^{n\times d}$ and $2 < p < \infty$. Let $\bfS$ be a random $\ell_p$ sampling matrix with sampling probabilities $q_i = \min\{1, 1/n + \bfsigma_i^p(\bfA) / \alpha\}$ for an oversampling parameter $\alpha$ set to
\begin{align*}
    \frac1\alpha &= O(p^2)\frS^p(\bfA)^{1-2/p}(l\log n)^{1-2/p}\log(pd) \log \frac{l\log n}{\eps} + O(p^2) l
\end{align*}
for
\[
    l = O\parens*{\log\frac1\delta + \log\log n + \log p + \log\frac{\frS^p(\bfA)}{\eps}}.
\]
Then, with probability at least $1-\delta$, simultaneously for all $\bfx\in\mathbb R^d$,
\[
    \norm*{\bfS\bfA\bfx}_p^p = (1\pm\eps)\norm*{\bfA\bfx}_p^p.
\]
Furthermore, with probability at least $1-\delta$, $\bfS$ samples
\[
    \frac{\frS^p(\bfA)^{2-2/p}}{\eps^2} \poly\parens*{\log n, \log\frac1\delta, p}
\]
rows.
\end{theorem}
\begin{proof}
Our approach is to bound
\[
    \E_{\bfS}\sup_{\norm*{\bfA\bfx}_p = 1} \abs*{\norm*{\bfS\bfA\bfx}_p^p - 1}^l
\]
for a large even integer $l$. Using \cref{lem:gp-reduction}, we first bound
\[
    \E_{\bfS}\sup_{\norm*{\bfA\bfx}_p = 1} \abs*{\norm*{\bfS\bfA\bfx}_p^p - 1}^l \leq (2\pi)^{l/2} \E_{\bfS}\E_{\bfg\sim\mathcal N(0,\bfI_n)}\sup_{\norm*{\bfA\bfx}_p = 1} \abs*{\sum_{i\in S} \bfg_i \abs*{[\bfS\bfA\bfx](i)}^p}^l
\]
where $S = \{i\in[n] : q_i < 1\}$. For simplicity of presentation, we assume $S = [n]$, which will not affect our proof.

By \cref{lem:recursive-sens-sampling}, there exists a matrix $\bfA'\in\mathbb R^{m_1\times d}$ with $m_1 = O(\frS^{2-2/p}\log(pd)^3)$ such that
\[
    \norm*{\bfA'\bfx}_p^p = (1\pm1/2)\norm*{\bfA\bfx}_p^p
\]
for all $\bfx\in\mathbb R^d$, and $\frS^p(\bfA') \leq O(1)\frS^p(\bfA)$. Then for $m_2 = O(m_1 + \frS^p(\bfA)\alpha^{-1})$, let $\bfA''\in\mathbb R^{m_2\times d}$ be the matrix given by \cref{lem:sens-flat} such that $\bfsigma_i^p(\bfA'') \leq \alpha$ for every $i\in[m_2]$ and $\norm*{\bfA''\bfx}_p = \norm*{\bfA'\bfx}_p$ for every $\bfx\in\mathbb R^d$. Now let
\[
    \bfA''' \coloneqq \begin{pmatrix}\bfA'' \\ \bfS\bfA\end{pmatrix}
\]
be the $(m_2 + n_\bfS)\times d$ matrix formed by the vertical concatenation of $\bfA''$ with $\bfS\bfA$, where $n_\bfS$ is the number of rows sampled by $\bfS$.

\paragraph{Sensitivity Bounds for $\bfA'''$.} We will first bound the $\ell_p$ sensitivities of $\bfA'''$. For any row $i$ corresponding to a row of $\bfA''$, the $\ell_p$ sensitivities are already bounded by $\alpha$, and furthermore, $\ell_p$ sensitivities can only decrease with row additions. For any row $i$ corresponding to a row of $\bfS\bfA$ that is sampled with probability $q_i < 1$, we have that
\[
    \frac{\abs{[\bfS\bfA\bfx](i)}^p}{\norm*{\bfA'''\bfx}_p^p} \leq \frac{\abs{[\bfS\bfA\bfx](i)}^p}{\norm*{\bfA''\bfx}_p^p} \leq 2\frac{\abs{[\bfS\bfA\bfx](i)}^p}{\norm*{\bfA\bfx}_p^p} \leq 2\alpha.
\]
By \cref{lem:sens-mon}, this immediately implies that the $\ell_2$ sensitivities, or the leverage scores, are also bounded by $2\alpha$.

\paragraph{Moment Bounds on Sampling Error.} We now fix a choice of $\bfS$, and define
\[
    F_\bfS \coloneqq \sup_{\norm*{\bfA\bfx}_p = 1} \abs*{\norm*{\bfS\bfA\bfx}_p^p - 1}
\]
Note that the event that $n_\bfS$ is at least
\[
    n_{\mathrm{thresh}} \coloneqq O(l\log n)\E[n_S] = O(l\log n) \alpha^{-1}\frS^p(\bfA),
\]
occurs with probability at most $\poly(n)^{-l}$ by Chernoff bounds over the randomness of $\bfS$, and
\[
    F_{\bfS}^l \leq \bracks*{1 + \sum_{i=1}^n \frac1{q_i}}^l \leq (n+1)^{2l},
\]
and thus this event contributes at most $\poly(n)^{-l}$ to the moment bound $\E F_{\bfS}^l$. Thus, we focus on bounding $\E F_\bfS^l$ conditioned on $n_\bfS \leq n_{\mathrm{thresh}}$. Now define
\[
    G_\bfS \coloneqq \sup_{\norm*{\bfA'''\bfx}_p = 1} \abs*{\sum_{i=1}^{m_2 + n_\bfS} \bfg_i\abs*{[\bfA'''\bfx](i)}^p}
\]
for $\bfg\sim\mathcal N(0,\bfI_{m_2 + n_\bfS})$. Then,
\[
    \norm*{\bfA'''\bfx}_p^p \leq (1 + 2 + F_{\bfS})\norm*{\bfA\bfx}_p^p
\]
so
\begin{equation}
\label{eq:Fsl-bound-p>2}
\begin{aligned}
    F_\bfS^l &\leq 2^l\sup_{\norm*{\bfA\bfx}_p = 1} \abs*{\sum_{i=1}^{m_2 + n_\bfS} \bfg_i \abs*{[\bfA'''\bfx](i)}^p}^l \\
    &\leq 2^l(1 + 2 + F_{\bfS})^l\sup_{\norm*{\bfA'''\bfx}_p = 1} \abs*{\sum_{i=1}^{m_2 + n_\bfS} \bfg_i \abs*{[\bfA'''\bfx](i)}^p}^l \\
    &\leq 2^{2l-1}(3^l+F_\bfS^l) G_\bfS^l.
\end{aligned}
\end{equation}
We then take expectations on both sides with respect to $\bfg\sim\mathcal N(0,\bfI_{m_2+n_\bfS})$, and bound the right hand side using \cref{lem:moment-bound}, which gives
\[
    \E_{\bfg\sim\mathcal N(0,\bfI_{m_2+n_\bfS})}G_\bfS^l \leq \parens*{2\mathcal E}^l\frac{\mathcal E}{\mathcal D} + O(\sqrt l \mathcal D)^l
\]
where $\mathcal E$ is the entropy integral and $\mathcal D = 4p\sigma^{1/2}$ is the diameter by \cref{lem:diam}. We have by \cref{lem:entropy-int-p>2} that
\begin{align*}
    \mathcal E &\leq O(p\tau^{1/2})\cdot (\sigma (m_2+n_\bfS))^{1/2-1/p} (\log(m_2+n_\bfS))^{1/2}\cdot \log\frac{p^2 d\sigma}{\tau} \\
    &\leq O(p\alpha^{1/2})\cdot (\alpha (m_2+n_\bfS))^{1/2-1/p} (\log(m_2+n_\bfS))^{1/2}\cdot \log(pd).
\end{align*}
Thus, conditioned on $n_\bfS \leq n_{\mathrm{thresh}}$, we have that
\begin{align*}
    \E_{\bfg\sim\mathcal N(0,\bfI_{m_2+n_\bfS})}G_\bfS^l &\leq \bracks*{O(p\alpha^{1/2})(\alpha(m_2 + n_{\mathrm{thresh}}))^{1/2-1/p}\parens*{\log (m_2 + n_{\mathrm{thresh}})}^{1/2}\log(pd)}^l + O(\sqrt l p\sqrt\alpha)^l \\
    &\leq \bracks*{O(p\alpha^{1-1/p})n_{\mathrm{thresh}}^{1/2-1/p}\parens*{\log n_{\mathrm{thresh}}}^{1/2}\log(pd)}^l + O(\sqrt l p\sqrt\alpha)^l \\
\end{align*}
Note that
\[
    \alpha^{1-1/p} n_{\mathrm{thresh}}^{1/2-1/p} = O(l\log n)^{1/2-1/p} \alpha^{1-1/p} (\alpha^{-1}\frS^p(\bfA))^{1/2-1/p} = O(l\log n)^{1/2-1/p} \alpha^{1/2} \frS^p(\bfA)^{1/2-1/p},
\]
which shows that
\[
    \E_{\bfg\sim\mathcal N(0,\bfI_{m_2+n_\bfS})} G_\bfS^l \leq \eps^l \delta
\]
due to our choice of $\alpha$ and $l$.

Now if we take conditional expectations on both sides of \eqref{eq:Fsl-bound-p>2} conditioned on the event $\mathcal F$ that $n_\bfS \leq n_{\mathrm{thresh}}$, then we have
\[
    \E[F_\bfS^l \mid \mathcal F] \leq 2^{2l-1}(3^l + \E[F_\bfS^l \mid \mathcal F])\eps^l \delta \leq (3^l + \E[F_\bfS^l \mid \mathcal F])(4\eps)^l \delta
\]
which means
\[
    \E[F_\bfS^l \mid \mathcal F] \leq \frac{(12\eps)^l \delta}{1 - (4\eps)^l\delta} \leq 2(12\eps)^l \delta
\]
for $(4\eps)^l\delta \leq 1/2$. We thus have
\[
    \E[F_\bfS^l] \leq \frac{(12\eps)^l \delta}{1 - (4\eps)^l\delta} \leq 2(12\eps)^l \delta + \poly(n)^{-l}
\]
altogether. Finally, we have by a Markov bound that
\[
    F_\bfS^l \leq 2(12\eps)^l + \frac1\delta \poly(n)^l \leq 3(12\eps)^l
\]
with probability at least $1-\delta$, which means that
\[
    F_\bfS \leq 3\cdot 12\eps = 36\eps
\]
with probability at least $1-\delta$. Rescaling $\eps$ by constant factors yields the claimed result.
\end{proof}

\subsection{Root Leverage Score Sampling, \texorpdfstring{$p<2$}{p<2}}

We start with a flattening lemma, which shows how to obtain an $\ell_p$ isometry that simultaneously flatten all $\ell_q$ sensitivities. 

\begin{lemma}[Flattening All Sensitivities]
\label{lem:flatten-all}
Let $1 \leq p < \infty$ and $\bfA\in\mathbb R^{n\times d}$. Let $0<\alpha<1$. Then, there exists $\bfA'\in\mathbb R^{m\times d}$ for $m = O(n\alpha^{-1})$ such that
\[
    \bfsigma_i^q(\bfA') \leq \alpha
\]
for every $i\in[m]$ and $1\leq q < \infty$. Furthermore, for any $1 \leq q < \infty$ and $\bfx\in\mathbb R^d$, we have that $\norm*{\bfA'\bfx}_q = \Theta(\alpha^{1/p-1/q})\norm*{\bfA\bfx}_q$.
\end{lemma}
\begin{proof}
Let $k \coloneqq \ceil*{1 / \alpha}$. Then, we construct $\bfA'\in\mathbb R^{m\times d}$ for $m = nk$ by replacing the $i$th row $\bfa_i$ of $\bfA$ for every $i\in[n]$ with $k$ copies of $\bfa / k^{1/p}$. Then, for every row $j\in[m]$ that is a copy of row $i\in[n]$, we have that
\[
    \bfsigma_j^q(\bfA) = \sup_{\bfA\bfx\neq 0}\frac{\abs*{[\bfA'\bfx](i)}^q}{\norm*{\bfA'\bfx}_q^q} \leq \sup_{\bfA\bfx\neq 0}\frac{\abs*{[k^{-1/p}\bfA\bfx](i)}^q}{k\cdot \abs*{[k^{-1/p}\bfA\bfx](i)}^q} \leq \frac1k\leq \alpha
\]
as desired. The second conclusion holds since
\[
    \norm*{\bfA'\bfx}_q^q = k\cdot k^{-q/p}\norm*{\bfA\bfx}_q^q = k^{1-q/p}\norm*{\bfA\bfx}_q^q.
\]
\end{proof}

\begin{theorem}[Root Leverage Score Sampling]
\label{thm:root-lev-sampling}
Let $\bfA\in\mathbb R^{n\times d}$ and let $1 \leq p < 2$. Let $0<\eps,\delta<1$. Let $\bfS$ be a random $\ell_p$ sampling matrix with sampling probabilities $q_i = \min\{1, \bftau_i(\bfA)^{p/2} / \alpha\}$ for an oversampling parameter $\alpha$ set to
\[
    \frac1\alpha = O(\eps^{-2})(\log d)^2\parens*{\frac{\log d}{2-p} + \log n + \log\frac1\delta}
\]
Then, with probability at least $1-\delta$, simultaneously for all $\bfx\in\mathbb R^d$,
\[
    \norm*{\bfS\bfA\bfx}_p^p = (1\pm\eps)\norm*{\bfA\bfx}_p^p.
\]
Furthermore, with probability at least $1-\delta$, $\bfS$ samples
\[
    \frac{n^{1-p/2}d^{p/2}}{\eps^2}\poly\parens*{\log n, \log\frac1\delta, \frac1{2-p}}
\]
rows.
\end{theorem}
\begin{proof}
Our approach is to bound
\[
    \E_{\bfS}\sup_{\norm*{\bfA\bfx}_p = 1} \abs*{\norm*{\bfS\bfA\bfx}_p^p - 1}^l
\]
for a large even integer $l$. Using \cref{lem:gp-reduction}, we first bound
\[
    \E_{\bfS}\sup_{\norm*{\bfA\bfx}_p = 1} \abs*{\norm*{\bfS\bfA\bfx}_p^p - 1}^l \leq (2\pi)^{l/2} \E_{\bfS}\E_{\bfg\sim\mathcal N(0,\bfI_n)}\sup_{\norm*{\bfA\bfx}_p = 1} \abs*{\sum_{i\in S} \bfg_i \abs*{[\bfS\bfA\bfx](i)}^p}^l
\]
where $S = \{i\in[n] : q_i < 1\}$. For simplicity of presentation, we assume $S = [n]$, which will not affect our proof.

By \cref{lem:flatten-all}, there exists a matrix $\bfA'\in\mathbb R^{m_1\times d}$ with $m_1 = O(n/\alpha)$ such that $\norm{\bfA'\bfx}_p = \norm{\bfA\bfx}_p$ and $\norm{\bfA'\bfx}_2 = \Theta(\alpha^{1/p-1/2})\norm{\bfA\bfx}_2$ for all $\bfx\in\mathbb R^d$, and $\bftau_i(\bfA) = \bfsigma_i^2(\bfA) \leq \alpha$ and $\bfsigma_i^p(\bfA) \leq \alpha$ for all $i\in[n]$. Now let
\[
    \bfA'' \coloneqq \begin{pmatrix}
    \bfA' \\ \bfS\bfA
    \end{pmatrix}
\]
be the $(m_1 + n_\bfS) \times d$ matrix formed by the vertical concatenation of $\bfA'$ with $\bfS\bfA$, where $n_\bfS$ is the number of rows sampled by $\bfS$. 

\paragraph{Leverage Score Bounds for $\bfA''$.}

We will first bound the leverage scores (or $\ell_2$ sensitivities) of $\bfA''$. For any row $i$ corresponding to a row of $\bfA'$, the $\ell_2$ sensitivities are already bounded by $\alpha$, and furthermore, $\ell_2$ sensitivities can clearly only decrease with row additions. For any row $i$ corresponding to a row of $\bfS\bfA$ that is sampled with probability $q_i < 1$, we have that
\[
    \frac{\abs*{[\bfS\bfA\bfx](i)}^2}{\norm*{\bfA''\bfx}_2^2} \leq \frac{\abs*{[\bfS\bfA\bfx](i)}^2}{\norm*{\bfA'\bfx}_2^2} = \frac{\abs*{[\bfS\bfA\bfx](i)}^2}{\Theta(\alpha^{2/p-1})\norm*{\bfA\bfx}_2^2} \leq \frac1{q_i^{2/p}}\frac{\abs*{[\bfA\bfx](i)}^2}{\Theta(\alpha^{2/p-1})\norm*{\bfA\bfx}_2^2} \leq \frac{\bftau_i(\bfA)}{\Theta(\alpha^{2/p-1}) q_i^{2/p}} = O(\alpha).
\]
Thus, we have that $\bftau_i(\bfA'') = \bfsigma_i^2(\bfA'') \leq O(\alpha)$ for every row $i$ of $\bfA''$. By monotonicity of max sensitivity \cref{lem:sens-mon}, we also have that $\bfsigma_i^p(\bfA) \leq O(\alpha)$.

\paragraph{Moment Bounds on the Sampling Error.}

We now fix a choice of $\bfS$, and define
\begin{align*}
F_{\bfS} &\coloneqq \sup_{\norm*{\bfA\bfx}_p = 1} \abs*{\norm*{\bfS\bfA\bfx}_p^p - 1} \\
G_\bfS &\coloneqq \sup_{\norm*{\bfA''\bfx}_p = 1} \abs*{\sum_{i=1}^{m_2 + n_\bfS} \bfg_i \abs*{[\bfA''\bfx](i)}^p}
\end{align*}
for $\bfg\sim\mathcal N(0,\bfI_{m_2 + n_\bfS})$. Then,
\[
    \norm*{\bfA''\bfx}_p^p \leq (1 + 2 + F_{\bfS})\norm*{\bfA\bfx}_p^p
\]
so
\begin{equation}
\label{eq:Fsl-bound-root-lev}
\begin{aligned}
    F_\bfS^l &\leq 2^l\sup_{\norm*{\bfA\bfx}_p = 1} \abs*{\sum_{i=1}^{m_1 + n_\bfS} \bfg_i \abs*{[\bfA''\bfx](i)}^p}^l \\
    &\leq 2^l(1 + 2 + F_{\bfS})^l\sup_{\norm*{\bfA''\bfx}_p = 1} \abs*{\sum_{i=1}^{m_1 + n_\bfS} \bfg_i \abs*{[\bfA''\bfx](i)}^p}^l \\
    &\leq 2^{2l-1}(3^l+F_\bfS^l) G_\bfS^l.
\end{aligned}
\end{equation}
We then take expectations on both sides with respect to $\bfg\sim\mathcal N(0,\bfI_{m_2+n_\bfS})$, and bound the right hand side using \cref{lem:moment-bound}, which gives
\[
    \E_{\bfg\sim\mathcal N(0,\bfI_{m_2+n_\bfS})}G_\bfS^l \leq \parens*{2\mathcal E}^l\frac{\mathcal E}{\mathcal D} + O(\sqrt l \mathcal D)^l
\]
where $\mathcal E$ is the entropy integral and $\mathcal D = 4\sigma^{1/2} \leq 4\alpha^{1/2}$ is the diameter by \cref{lem:diam}. We have by \cref{lem:entropy-int-p<2} that
\begin{align*}
    \mathcal E &\leq O(\tau^{1/2})\parens*{\frac{\log d}{2-p} + \log(m_1+n_\bfS)}^{1/2}\log\frac{d\sigma}{\tau} \\
    &\leq O(\alpha^{1/2})\parens*{\frac{\log d}{2-p} + \log n}^{1/2}\log d
\end{align*}
By our choice of $\alpha$ and $l$, we have
\[
    \E_{\bfg\sim\mathcal N(0,\bfI_{m_2+n_\bfS})} G_\bfS^l \leq \eps^l \delta.
\]

Now if we take conditional expectations on both sides of \eqref{eq:Fsl-bound-root-lev}, then we have
\[
    \E[F_\bfS^l] \leq 2^{2l-1}(3^l + \E[F_\bfS^l])\eps^l \delta \leq (3^l + \E[F_\bfS^l])(4\eps)^l \delta
\]
which means
\[
    \E[F_\bfS^l] \leq \frac{(12\eps)^l \delta}{1 - (4\eps)^l\delta} \leq 2(12\eps)^l \delta
\]
for $(4\eps)^l\delta \leq 1/2$. Finally, we have by a Markov bound that $F_\bfS^l \leq 2(12\eps)^l$ with probability at least $1-\delta$, which means that $F_\bfS \leq 2\cdot 12\eps = 24\eps$ with probability at least $1-\delta$. Rescaling $\eps$ by constant factors yields the claimed sampling error bound.

Note that the expected number of rows sampled is at most
\[
    \frac1\alpha\sum_{i=1}^n \bftau_i(\bfA)^{p/2} \leq \frac1\alpha n^{1-p/2}\parens*{\sum_{i=1}^n \bftau_i(\bfA)}^{p/2} = \frac1\alpha n^{1-p/2}d^{p/2}
\]
by H\"older's inequality. This implies the bound on number of sampled rows by Chernoff bounds.
\end{proof}

Finally, we show that by recursively applying \cref{thm:root-lev-sampling}, we can reduce the number of rows to roughly $d / \eps^{4/p}$. To bound the size, we need to solve a recursion for an upper bound on the number of rows. This is given by the following:

\begin{lemma}[Lemma 6.12, \cite{MMWY2022}]
\label{lem:recurrence}
Suppose $(a_i)_{i=0}^\infty$ satisfies the recurrence $a_{i+1} = \lambda a_i + b$ for some $b>0$ and $\lambda\in(0,1)$. Then,
\[
    a_i = \frac1{1-\lambda}\parens*{b - \lambda^i(b - (1-\lambda)a_0)}.
\]
\end{lemma}

This gives the following

\begin{theorem}[Recursive Root Leverage Score Sampling]
\label{thm:recursive-root-lev-sampling}
Let $\bfA\in\mathbb R^{n\times d}$ and let $1 \leq p < 2$. Let $0<\eps,\delta<1$. Let $\bfS$ be the result of recursively applying \cref{thm:root-lev-sampling} with failure probability $\delta/\Theta(\log\log n)$ and accuracy $\eps / \Theta(\log\log n)$ recursively until the number of rows is at most
\[
    m = \frac{d}{\eps^{4/p}}\poly\parens*{\log n, \log\frac1\delta, \log\frac1{2-p}}
\]
Then, with probability at least $1-\delta$, simultaneously for all $\bfx\in\mathbb R^d$,
\[
    \norm*{\bfS\bfA\bfx}_p^p = (1\pm\eps)\norm*{\bfA\bfx}_p^p.
\]
\end{theorem}
\begin{proof}
We apply \cref{thm:root-lev-sampling} with failure probability $\delta/\Theta(\log\log n)$ and accuracy $\eps / \Theta(\log\log n)$ recursively for at most $R = O(\log\log n)$ rounds, until the number of rows is at most the claimed bound. By a union bound, we succeed at achieving $\eps /\Theta(\log\log n)$ sampling error and sampling bound on all $R$ rounds, that is, for any number of rows $m$, we reduce the number of rows to at most
\[
    \frac{m^{1-p/2}d^{p/2}}{\eps^2} \poly\parens*{\log n,\log\frac1\delta,\frac1{2-p}}.
\]
We then apply the recurrence lemma \cref{lem:recurrence} on the logarithm of the above bound, so $\lambda = (1-p/2)$ and
\[
    b = \log\parens*{\frac{d^{p/2}}{\eps^2} \poly\parens*{\log n,\log\frac1\delta,\frac1{2-p}}}.
\]
Then, after $i = O(\log\log n)$ iterations, our log row count upper bound is
\[
    a_i = \log m_i = \frac2p \parens*{b - (1-p/2)^i (b - (p/2)\log n)} \leq \frac2p \parens*{b + (1-p/2)^i (p/2)\log n} \leq \frac2p (b+1)
\]
or
\[
    m_i \leq O\parens*{\frac{d^{p/2}}{\eps^2} \poly\parens*{\log n,\log\frac1\delta,\frac1{2-p}}}^{2/p} = \frac{d}{\eps^{4/p}} \poly\parens*{\log n,\log\frac1\delta,\frac1{2-p}}
\]
\end{proof}

\subsection{Leverage Score + \texorpdfstring{$\ell_p$}{lp} Sensitivity Sampling, \texorpdfstring{$p>2$}{p>2}}

We start with a flattening lemma for $p>2$, which shows how to slightly flatten leverage scores while preserving $\ell_p$ norms, and with only a small blow up in the number of rows.

\begin{lemma}[Flattening $\ell_p$ Sensitivities and Leverage Scores]
\label{lem:flat-sens-lev}
Let $\bfA\in\mathbb R^{n\times d}$ and let $2 < p < \infty$. Let $C\geq 1$. Then, there exists $\bfA'\in\mathbb R^{m\times d}$ for $m\leq (1+1/C)n$ such that $\norm*{\bfA\bfx}_p = \norm*{\bfA'\bfx}_p$ and $\norm*{\bfA'\bfx}_2 \geq \norm*{\bfA\bfx}_2$ for every $\bfx\in\mathbb R^d$ and satisfies
\begin{align*}
    \max_{i\in[m]}\bfsigma_i^p(\bfA') &\leq \max_{i\in[n]}\bfsigma_i^p(\bfA) \\
    \max_{i\in[m]}\bftau_i(\bfA') &\leq  (Cd/n)^{2/p}\max_{i\in[m]}\bftau_i(\bfA')^{1-2/p} \\
    \frS^p(\bfA') &= \frS^p(\bfA)
\end{align*}
\end{lemma}
\begin{proof}
The idea roughly follows \cref{lem:sens-flat}, except that we split rows that have large leverage score, rather than rows that have large $\ell_p$ sensitivity. For each $i\in[n]$, let $k_i = \ceil*{\bftau_i(\bfA) / (Cd/n)}$. Then for each $i\in[n]$ such that $k_i > 1$, we replace $i$th row $\bfa_i$ of $\bfA$ by $k_i$ copies of $\bfa_i / k^{1/p}$. Clearly, we have that $\norm*{\bfA'\bfx}_p = \norm*{\bfA\bfx}_p$ for every $\bfx\in\mathbb R^d$. Furthermore, the number of rows added is at most
\[
    \sum_{i : \bftau_i(\bfA) \geq Cd / n}\ceil*{\frac{\bftau_i(\bfA)}{Cd/n}} - 1 \leq \sum_{i : \bftau_i(\bfA) \geq Cd / n}\frac{\bftau_i(\bfA)}{Cd/n} = \frac{d}{Cd/n} = \frac{n}{C}.
\]
Next, note that for every $\bfx\in\mathbb R^d$,
\[
    k_i \cdot \abs*{k_i^{-2/p}[\bfA\bfx](i)}^2 \geq \abs*{[\bfA\bfx](i)}^2
\]
since $k_i \geq 1$, so we have that $\norm*{\bfA'\bfx}_2 \geq \norm*{\bfA\bfx}_2$. Then, for any row $j\in[m]$ that is a copy of row $i\in[n]$ of $\bfA$, we have that
\[
    \bftau_j(\bfA') = \sup_{\bfA\bfx\neq 0}\frac{\abs*{[\bfA'\bfx](i)}^2}{\norm*{\bfA'\bfx}_2^2} \leq \sup_{\bfA\bfx\neq 0}\frac{k_i^{-2/p}\abs*{[\bfA\bfx](i)}^2}{\norm*{\bfA\bfx}_2^2} \leq \frac{(Cd/n)^{2/p}}{\bftau_i(\bfA)^{2/p}}\bftau_i(\bfA) = (Cd/n)^{2/p} \bftau_i(\bfA)^{1-2/p}.
\]
Finally, it is clear that the $\ell_p$ sensitivities can only decrease and that the total $\ell_p$ sensitivity is preserved.
\end{proof}

Now using \cref{lem:flat-sens-lev}, we first obtain a construction of a small $\ell_p$ approximate isometry in a way analogous to \cref{lem:recursive-sens-sampling}.

\begin{theorem}[Recursive Leverage Score + $\ell_p$ Sensitivity Sampling]
\label{thm:recursive-sens-lev-sampling}
Let $\bfA\in\mathbb R^{n\times d}$ and $2 < p < \infty$. Let $0<\eps,\delta<1$. Then, there exists an efficient algorithm producing a matrix $\bfS\in\mathbb R^{m\times n}$ for 
\[
    m = O(p^2) \frac{d^{2/p}\frS^p(\bfA)^{2-4/p}}{\eps^2}\parens*{\log\frac{pd}{\delta}}^2\log \frac{pd}{\eps}.
\]
such that
\[
    \norm*{\bfS\bfA\bfx}_p^p = (1\pm\eps)\norm*{\bfA\bfx}_p^p
\]
for every $\bfx\in\mathbb R^d$ and $\frS^p(\bfA') \leq (1+O(\eps))\frS^p(\bfA)$.
\end{theorem}
\begin{proof}
Our proof is almost identical to \cref{lem:recursive-sens-sampling}, but we will use high probability versions of the results, as we will directly use the recursive sampling procedure algorithmically. A similar algorithmic recursive sampling procedure is considered in \cite{MMWY2022}. We first replace $\bfA'$ in \cref{lem:recursive-sens-sampling} with the matrix formed by first applying \cref{lem:sens-flat} with $C = 4$, and then applying \cref{lem:flat-sens-lev} with $C = 4$. The resulting $\bfA'$ has at most $(5/4)^2 n = (25/16) n$ rows, preserves $\ell_p$ norms and $\ell_p$ total sensitivity, and has
\begin{equation}\label{eq:sens-lev-bound}
\begin{aligned}
    \bfsigma_i^p(\bfA') &\leq O(1)\frac{\frS^p(\bfA)}{n} \\
    \bftau_i(\bfA') &\leq O(1)\parens*{\frac{d}{n}}^{2/p}\parens*{\frac{\frS^p(\bfA)}{n}}^{1-2/p} = O(1)\frac{d^{2/p}\frS^p(\bfA)^{1-2/p}}{n}
\end{aligned}
\end{equation}
We then again consider sampling half of the rows of $\bfA'$ via an $\ell_p$ sampling matrix $\bfS$ with $q_i = 1/2$. By \cref{lem:gp-reduction} and \cref{thm:dudley-tail}, we then have that 
\[
    \Pr\braces*{\sup_{\norm*{\bfA'\bfx}_p = 1}\abs*{\norm*{\bfS\bfA'\bfx}_p^p-1} \leq O(p\tau^{1/2})\cdot(\sigma n)^{1/2-1/p}(\log n)^{1/2}\cdot \log \frac{p^2 d \sigma}{\tau} + 4p\sigma^{1/2}\cdot z} \geq 1 - 2\exp(z^2)
\]
where $\tau$ is an upper bound on the leverage scores of $\bfA'$ and $\sigma$ is an upper bound on the $\ell_p$ sensitivities of $\bfA'$. These are bounded by \eqref{eq:sens-lev-bound}, and thus applying these bounds and setting $z = O(\log(n/\delta))$ gives
\[
    \Pr\braces*{\sup_{\norm*{\bfA'\bfx}_p = 1}\abs*{\norm*{\bfS\bfA'\bfx}_p^p-1} \leq O(p)\frac{d^{1/p}\frS^{p}(\bfA')^{1-2/p}}{\sqrt n}\parens*{(\log n)^{1/2}\cdot \log(pd) + \log\frac{\log n}{\delta}}} \geq 1 - \frac{\delta}{\poly\log n}.
\]
By a union bound, the same bound holds for the first $O(\log n)$ recursive calls to our recursive sampling algorithm, up to a different $\poly\log n$ factor in the denominator of the failure rate bound. Furthermore, by Chernoff bounds, we have that the number $n_\bfS$ of rows sampled by $\bfS$
\[
    \Pr\braces*{n_\bfS \leq (25/32)(5/4) n < 0.98 n} \geq 1 - \exp\parens*{-\frac13\cdot\frac1{16}\cdot\frac{25}{32}n} \geq 1 - \frac{\delta}{\poly\log n}
\]
as long as $n \geq C\log\frac1\delta$ for a sufficiently large constant $C$. By a union bound, the same bound also holds for the first $O(\log n)$ recursive calls to our sampling algorithm. In this case, our sampling process terminates in at most $O(\log n)$ rounds, so both the bound on $n_\bfS$ and the sampling error bound hold for all $O(\log n)$ rounds. 

We now apply the same reasoning as \cref{lem:recursive-sens-sampling} to bound the total sampling error. Let $\bfA_l$ denote the $n_l\times d$ sampled matrix after $l$ rounds of the recursive sampling procedure, and let
\[
    \eps_{\bfA_l} = O(p)\frac{d^{1/p}\frS^{p}(\bfA_l)^{1-2/p}}{\sqrt n_l}\parens*{(\log n_l)^{1/2}\cdot \log(pd) + \log\frac{\log n_l}{\delta}}.
\]
Then,
\begin{align*}
    \eps_{\bfA_{l+1}} &= O(p) \frac{d^{1/p}\frS^p(\bfA_{l+1})^{1-2/p}}{\sqrt n_{l+1}}\parens*{(\log n_{l+1})^{1/2}\log(pd) + \log\frac{\log n_{l+1}}{\delta}} \\
    &\geq (1-O(\eps_{\bfA_l})) O(p) \frac{d^{1/p}\frS^p(\bfA_{l})^{1-2/p}}{\sqrt n_{l+1}}\parens*{(\log n_{l+1})^{1/2}\log(pd) + \log\frac{\log n_{l+1}}{\delta}} \\
    &\geq \sqrt{\frac{1}{0.98}} (1-O(\eps_{\bfA_l})) O(p) \frac{d^{1/p}\frS^p(\bfA_{l})^{1-2/p}}{\sqrt n_{l}}\parens*{(\log n_{l})^{1/2}\log(pd) + \log\frac{\log n_{l}}{\delta}} \\
    &\geq \frac{101}{100} \cdot \eps_{\bfA_l}
\end{align*}
so the sum of the $\eps_{\bfA_l}$ are dominated by the last $\eps_{\bfA_l}$, up to a constant factor. Now let $L$ be the smallest integer $l$ such that $\eps_{\bfA_l} \leq \eps$. Then, we have that
\[
    \frS^p(\bfA_L) \leq (1+O(\eps))\frS^p(\bfA)
\]
and thus
\[
    \norm*{\bfA_L\bfx}_p^p = (1\pm O(\eps))\norm*{\bfA\bfx}_p^p
\]
for every $\bfx\in\mathbb R^d$. Furthermore, $n_L$ satisfies
\[
    \eps = O(p) \frac{d^{1/p}\frS^p(\bfA)^{1-2/p}}{\sqrt n_L}\parens*{(\log n_{L})^{1/2}\log(pd) + \log\frac{\log n_{L}}{\delta}}
\]
or
\[
    n_L = O(p^2) \frac{d^{2/p}\frS^p(\bfA)^{2-4/p}}{\eps^2}\parens*{\log\frac{pd}{\delta}}^2\log \frac{pd}{\eps}.
\]
\end{proof}